\documentclass[11pt]{amsart}

\usepackage{amsmath, amsthm, amssymb, amsfonts, enumerate}
\usepackage[colorlinks=true,linkcolor=blue,urlcolor=blue]{hyperref}
\usepackage{dsfont}
\usepackage{color}
\usepackage{geometry}
\usepackage{soul}

\usepackage{graphicx}
\usepackage{caption}
\usepackage{subcaption}

\def \cC{\mathcal{C}}

\def \cE{\mathcal{E}}
\def \cS{\mathcal{S}}

\def \cF{\mathcal{F}}

\def \cI{\mathcal{I}}
\def \cL{\mathcal{L}}

\def \cO{\mathcal{O}}

\def \P{\mathsf P}
\def \PP{\widehat{\mathsf P}}
\def \Q{\mathsf Q}
\def \E{\mathsf E}
\def \EE{\widehat{\mathsf E}}
\def \EQ{\mathsf E^{\mathsf Q}}
\def \R{\mathbb{R}}
\def \XX{\widetilde{X}}
\def \WW{\widetilde{W}}
\def \ud{\textrm{d}}

\newcommand{\eps}{\varepsilon}

\newtheorem{theorem}{Theorem}[section]
\newtheorem{lemma}[theorem]{Lemma}
\newtheorem{corollary}[theorem]{Corollary}
\newtheorem{proposition}[theorem]{Proposition}

\newtheorem{remark}[theorem]{Remark}

\title[Participating policies with surrender option]{An analytical study of participating policies \\ with minimum rate guarantee\\ and surrender option}
\author[M.B.~Chiarolla, T.~De Angelis, G.~Stabile]{Maria B.~Chiarolla \and Tiziano De Angelis \and Gabriele Stabile}
\subjclass[2010]{91G80, 62P05, 60G40, 35R35; {\em JEL Classification}. G22}
\keywords{Participating policies, minimum rate guarantee, surrender option, solvency requirement, optimal stopping, free-boundary problems, stop-loss boundary, too-good-to-persist boundary}
\address{M.B.~Chiarolla: Dipartimento di Scienze dell'Economia, Centro Ecotekne, Universit\`a del Salento, Lecce, Italy.}
\email{\href{mailto:maria.chiarolla@unisalento.it}{maria.chiarolla@unisalento.it}}
\address{T.~De Angelis: School of Management and Economics (Dept.\ ESOMAS) University of Torino and Collegio Carlo Alberto, Torino, Italy.}
\email{\href{mailto:tiziano.deangelis@unito.it}{tiziano.deangelis@unito.it}}
\address{G.~Stabile: Dipartimento di Metodi e Modelli per l'Economia, il Territorio e la Finanza, Sapienza-Universit\`{a} di Roma, Roma, Italy}
\email{\href{mailto:gabriele.stabile@uniroma1.it}{gabriele.stabile@uniroma1.it}}
\date{\today}

\numberwithin{equation}{section}

\begin{document}

\begin{abstract}
We perform a detailed theoretical study of the value of a class of participating policies with four key features: $(i)$ the policyholder is guaranteed a minimum interest rate on the policy reserve; $(ii)$ the contract can be terminated by the holder at any time until maturity (surrender option); $(iii)$ at the maturity (or upon surrender) a bonus can be credited to the holder if the portfolio backing the policy outperforms the current policy reserve; $(iv)$ due to solvency requirements the contract ends if the value of the underlying portfolio of assets falls below the policy reserve.

Our analysis is probabilistic and it relies on optimal stopping and free boundary theory. We find a structure of the optimal surrender strategy which was undetected by previous (mostly numerical) studies on the same topic. Optimal surrender of the contract is triggered by  two `stop-loss' boundaries and by a `too-good-to-persist' boundary (in the language of \cite{EV20}). 
Financial implications of this strategy are discussed in detail and supported by extensive numerical experiments.
\end{abstract}

\maketitle

\section{Introduction}

Participating Policies with Minimum Rate guarantee are insurance contracts, appealing predominantly to individuals during their working lives, as a form of low-risk financial investment. The subscriber of a participating policy (policyholder) pays a premium (either single or periodic) which is used by the insurance company to set up a so-called policy reserve for the policyholder. The policy reserve is linked to a portfolio of assets held by the insurance company and it accrues interest tracking the performance of such portfolio (the details of the contract are illustrated in Section \ref{sec:setup}). The minimum rate guarantee is in the form of a minimum interest rate paid by the insurance company towards the policy reserve irrespective of the performance of the portfolio backing the policy (this rate is usually lower than the risk-free rate). In the absence of any further contract specifications, the policy terminates at a given maturity, at which the policyholder receives an amount equal to the value of the reserve, plus a {\em bonus}, if the current value of the portfolio is sufficiently high relative to the reserve.

The contract may incur early termination. That happens if the value of the portfolio backing the policy is not sufficient to cover the policy reserve. In that case we say that the insurance company fails to meet the solvency requirements on the participating policy and the policyholder receives the value of the policy reserve at that time. More interestingly, early termination of the contract may be an {\em embedded option} in the policy specification. Indeed, along with the standard contract, the policyholder can buy the right to an early cancellation of the policy, the so-called {\em surrender option} (SO). If the policyholder exercises the surrender option prior to the maturity of the contract, at that time she receives the value of the policy reserve, plus the above mentioned {\em bonus}.

While surrender options share similarities with financial options of American type, due to their early exercise feature, they are actually rather different in nature. Indeed, the presence of a surrender option, embedded in a participating policy, changes the structure of the whole contract. As a consequence, the price of embedded options is normally defined as the difference between the value of a policy which includes the option and the value of a policy which does not include the option (see \eqref{eq:pricing} for a mathematical expression).

\textit{Participating Policies with Surrender Option} (PPSO) have been studied extensively in the academic literature. This paper contributes to that strand of the literature which assumes that the policyholder is fully rational and the surrender option is exercised optimally from a financial perspective. Other papers analyse PPSO in which surrender occurs as a randomised event (see Cheng and Li \cite{CHENG}) or without assuming rational behaviour of the policyholder (see Nolte and Schneider \cite{NOLTE}).
Several papers adopt a numerical approach to analyse PPSO without solvency requirements for the insurance company (see Andreatta and Corradin \cite{ANDREATTA}, Bacinello \cite{BACINELLO1}, Bacinello, Biffis and Millossovich \cite{BACINELLO3}, Grosen and J{\o}rgensen \cite{GROSEN} among others).
Chu and Kwok \cite{CHUKWOK} provide an analytical approximation for the  price of a participating policy without taking into account the surrender option. Finally, Siu \cite{SIU} considers the fair valuation of a PPSO when the market value of the portfolio backing the policy is modelled by a Markov-modulated Geometric Brownian Motion. In \cite{SIU} the author approximates the solution of a free boundary problem by a system of second-order piecewise linear ordinary differential equations.

In this paper we develop a fully theoretical analysis of participating policies with minimum rate guarantee, embedded surrender option and early termination due to solvency requirements. Following the example of other papers on this topic (see, e.g., Chu and Kwok \cite{CHUKWOK}, Fard and Siu \cite{FS13}, Grosen and J{\o}rgensen \cite{GROSEN}, Siu \cite{SIU}), we focus purely on the financial aspects of the policy and ignore the demographic risk, in the sense that the contract does not account for a possible demise of the policyholder. From the point of view of applications we may imagine that there are multiple beneficiaries of the policy, so that the demographic risk is negligible. Moreover, it is well-known (see, e.g., Cheng and Li \cite{CHENG}, Stabile \cite{Stabile}) that the assumption of a constant force of mortality, independent of the financial market,
results in a shift in the discount rate adopted for pricing; this does not affect the methods we employ and the qualitative outcomes of our work. The study becomes substantially more involved if one considers a time-dependent mortality force (as, e.g., in De Angelis and Stabile \cite{DeAS2}) or worse a stochastic mortality. We leave these extensions for future work.

Our main contributions are: (i) the analytical study of the pricing formula for the PPSO, (ii) a characterisation of the optimal exercise strategy for the surrender option, in terms of an optimal exercise boundary, and (iii) an extensive numerical analysis of the option value and the surrender boundary. The arbitrage-free price of the policy is obtained as the value function of a suitable finite-time horizon optimal stopping problem, on a two-dimensional degenerate diffusion which lives in an orthant of the plane and is absorbed upon leaving the orthant. The diffusive coordinate of such process models the dynamics of the portfolio backing the policy, while the other coordinate represents the dynamics of the policy reserve. After a suitable transformation, the dynamics is reduced to a one dimensional diffusion in the form of a stochastic differential equation (SDE) with absorption upon hitting zero. This process corresponds to the so-called bonus distribution rate of the contract, which will be introduced in more detail in the next section. We are then led to consider a state process $(t,X_t)\in [0,T]\times\R_+$ which is absorbed if $X_t$ reaches zero, so that our state space is the $(t,x)$-strip with $t\in[0,T]$ and $x\ge 0$.

The optimal stopping problem poses a number of challenges: it is set on a finite-time horizon, hence it is not amenable to explicit solutions using the associated free boundary problem (which is indeed parabolic); the SDE that describes the stochastic process does not admit an explicit solution, so that numerous tricks often used in optimal stopping problems, and relying on an explicit dependence of the process on its initial value, are not applicable (see, e.g., the American put problem in Peskir and Shiryaev \cite{PS}); the stopping payoff is independent of time but, as a function of $x$, it is convex with discontinuous first derivative.

The combination of the above ingredients produces a very peculiar shape of the optimal stopping region, which we derive from a detailed analysis of the value function.  We observe that the stopping region $\cS$, i.e., the points $(t,x)$ at which the policyholder should instantly surrender the contract, is not connected in the $x$-variable, for each value of $t$ given and fixed. Instead, $\cS$ may have two connected components for each $t\in[0,T]$ (see Figure \ref{fig1}), corresponding to two distinct `stop-loss' boundaries and a `too-good-to-persist' boundary (following \cite{EV20}, here we say that a boundary is `stop-loss' or `too-good-to-persist' if $\cS$ can be locally represented as a set $\{(t,x):x\le b(t)\}$ or $\{(t,x):x\ge b(t)\}$, respectively, for some $b$). This result was not observed in prior work on the same model, where a numerical approach to the problem could not detect this unusual feature (see, e.g., Siu \cite{SIU}). As it turns out, the shape of the stopping set is closely related to the bonus mechanism included in the PPSO and it has fine implications on the optimal exercise of the surrender option. We will elaborate further on this point in Section \ref{sec:num}, once the mathematical details have been laid out more clearly.

The stopping set is connected in the $t$-variable, for each given $x\ge 0$. This leads naturally to consider an optimal stopping boundary as a function of $x$ (rather than as a function of $t$, as in the vast majority of papers in the area). We obtain a wealth of fine properties of the map $x\mapsto c(x)$ on $[0,\infty)$, which are of independent mathematical interest for the probabilistic theory of free boundary problems. Indeed we show that $c(\,\cdot\,)$ is continuous on $(0,\infty)$ and piecewise monotonic, with two {\em strictly increasing} portions and a {\em strictly decreasing} one. It is important to remark that questions of continuity of the optimal boundary $x\mapsto c(x)$ are much harder to address than in the more canonical setting of time dependent boundaries $t\mapsto b(t)$. Here we resolve the issue in Theorem \ref{prop:cont}, by providing a probabilistic proof which is new in the literature and makes use of suitably constructed reflecting diffusions. Our proof provides a conceptually simple way to show (in more general examples) that time dependent optimal boundaries $t\mapsto b(t)$ cannot exhibit flat stretches, unless the smooth-fit property fails.

The rest of the paper is organised as follows. In Section \ref{sec:setup} we set up the model in a rigorous mathematical framework and then we state our main results in Section \ref{sec:mainres} (Theorems \ref{thm:v} and \ref{thm:boundary}). In particular, in Section \ref{sec:num} we obtain numerical illustrations of the value function, the stopping set and the related sensitivity analysis, accompanied by a financial interpretation. Section \ref{sec:value} contains preliminary technical results on the continuity and monotonicity of the value function. In Section \ref{sec:freeb} we analyse in detail the free boundary problem associated with the PPSO and we prove Theorems \ref{thm:v} and \ref{thm:boundary} stated in Section \ref{sec:mainres}. Section \ref{sec:fees} extends our framework to include management fees in the valuation of the PPSO. The paper is completed by a short technical appendix.


\section{Actuarial model and problem formulation}\label{sec:setup}

In this section we provide a mathematical description of the price of a PPSO in a complete market, under a risk-neutral probability measure. We align our setup and part of our notations to those already used in other papers on this topic as, e.g.,~Chu and Kwok \cite{CHUKWOK}, Grosen and J{\o}rgensen \cite{GROSEN} and Siu \cite{SIU}.

Given $T>0$, we consider a market with finite time horizon $[0,T]$ on a complete probability space $(\Omega, \cF, (\cF_t)_{t\in[0,T]}, \Q)$ that carries a
one-dimensional Brownian motion $\WW:=(\WW_{t})_{t\in[0,T]}$. With no loss of generality we assume that the filtration $(\cF_t)_{t\in[0,T]}$ is generated by the Brownian motion and it is completed with $\Q$-null sets. Our market is complete and $\Q$ is the risk-neutral probability measure.

An investor can purchase a PPSO at time zero by making a lump payment $V_0$ to an insurance company. In return, the insurer invests an amount $R_0$ into a financial portfolio and commits the company to credit interests to the policyholder's {\em policy reserve} according to a mechanism that will be described below. Thanks to the {\em surrender option} embedded in the contract, the policyholder has the right to withdraw her investment at any time prior to the policy's maturity $T$. In this case, she receives the so-called {\em intrinsic value} of the policy.

\subsection{The policy reserve}
First we describe the rate at which the amount $R_0$ invested by the insurer accrues interest, based on the performance of the portfolio backing the policy (the \textit{reference portfolio}). We let $A:=(A_{t})_{t\in[0,T]}$ be the process denoting the value of such portfolio and assume that $A$ evolves as a geometric Brownian motion under $\Q$; that is,
\begin{align}\label{dynA}
\left \{
\begin{array}{l}
\ud A_{t}= A_{t}\big( r  \ud t + \sigma  \ud \WW_{t}\big),  \\
A_{0} = a_0,
\end{array} \right.
\end{align}
where $r$, $\sigma$ and $a_0$ are positive constants and $r$ is the risk free rate.

During the lifetime of the policy, the \emph{policy reserve} is denoted by $R:=(R_{t})_{t\in[0,T]}$ and accrues interest based on a two-layer mechanism. First, the insurance company guarantees a minimum fixed interest rate, which we denote by $r^G$ and, in line with financial practice, we assume
\begin{align}\label{def:rg}
r^G\in(0,r).
\end{align}
Second, at times when the portfolio performs particularly well, the policyholder participates in the returns. In particular, we define the so-called \emph{bonus reserve} $B_t:=A_t-R_t$ and, as in \cite{CHUKWOK}, \cite{SIU}, \cite{SFH} and \cite{FS13}, we consider a {\em bonus distribution rate} (BDR) of the form
\begin{align}\label{eq:BDR}
\ln\big(1+B_t/R_t\big)=\ln\big(A_t/R_t\big).
\end{align}

The BDR measures the performance of the portfolio against the performance of the policy reserve. The insurance company compares the BDR to a constant, long-term target $\beta>0$, known as \emph{target buffer ratio}. If the BDR exceeds the target buffer ratio, a proportion $\delta>0$ of the excess is shared with the policyholder.

Combining the minimum guaranteed interest rate with the bonus rate gives the instantaneous rate of interest on the policy reserve, that is
\begin{equation}\label{scheme}
c(A_t,R_t) = \delta \bigg ( \ln \Big( \frac{A_t}{R_t}\Big) - \beta \bigg ) \vee r^G.
\end{equation}

It follows that the policy reserve $R$ evolves under $\Q$ according to the dynamics
\begin{align}\label{dynP}
\left \{ \begin{array}{l}
\ud R_{t}= c(A_{t},R_{t}) R_{t} \ud t,  \\
R_{0} = \alpha\,a_0,
\end{array} \right.
\end{align}
where $\alpha\in(0,1)$ is fixed by the insurer. Hence, the initial reserve $R_0$ covers $\alpha$ shares of the reference portfolio.

\begin{remark}
Notice that in the specification of the bonus mechanism in \eqref{scheme} we may equivalently consider $\ln(\alpha A_t/R_t)$ instead of $\ln(A_t/R_t)$. This would emphasise that the policyholder only receives a bonus proportional to her share of the portfolio backing the policy. From the mathematical point of view, of course there is no difference since $\ln(\alpha A_t/R_t)=\ln \alpha+\ln(A_t/R_t)$ and the additional term, $\ln\alpha$, is absorbed in the specification of the target buffer ratio $\beta>0$.
\end{remark}

\subsection{Intrinsic value of the policy and arbitrage-free price}
Next we describe the so-called intrinsic value of the policy, which is the value that the policyholder receives either at the maturity of the policy or at an earlier time, should she decide to exercise the surrender option.

The intrinsic value is equal to the policy reserve plus a bonus component. The latter, is activated when the value of the policyholder's shares in the portfolio $A$ exceeds the current value of the policy reserve; that is, when ${\alpha A_t}>R_{t}$. In this case the policyholder receives a bonus fraction $\gamma$ of the surplus of her $\alpha$-share.

From the mathematical point of view, the intrinsic value of the policy may be written as
\begin{align}\label{gdef}
g(A_{t},R_{t}) :=R_{t} + \gamma \left[\alpha A_{t} - R_{t}\right]^{+},
\end{align}
where $[x]^+:=\max\{x,0\}$ and $\gamma\in(0,1)$ is the so-called \emph{participation coefficient}.

The model also takes into account that the company may fail to meet the solvency requirement at any time before $T$. In fact, denoting by $\tau^{\dagger}$ the stopping time ({\em insolvency time})
\begin{align}\label{taumorto}
\tau^\dagger:=\inf\{t\ge0\,:\,A_{t} \le R_{t}\},
\end{align}
the company's solvency requirement is satisfied for $t<\tau^\dagger$. In the event of $\tau^\dagger < T$ the policy is liquidated and the policyholder receives (cf.~\eqref{gdef})
\begin{align*}
g(A_{\tau^{\dagger}}, R_{\tau^{\dagger}}) = R_{\tau^{\dagger}},
\end{align*}
i.e., the policy reserve value.

Finally, we can define $V_0$, the arbitrage-free price of the PPSO at time zero. Notice that $V_0=V_0(\alpha)$, in the sense that the contract is specified by indicating the portion $\alpha$ of the portfolio which backs the policy.  Recalling $\eqref{dynA}$, \eqref{dynP}, \eqref{gdef} and \eqref{taumorto}, we have
\begin{equation}
\label{value1}
V_0 = \sup_{0\leq\tau\leq T}\EQ\Big[e^{-r\, (\tau\wedge{\tau^\dagger})} g(A_{\tau\wedge\tau^\dagger}, R_{\tau\wedge\tau^\dagger})\Big],
\end{equation}
where $\EQ$ is the expectation under the measure $\Q$ and the supremum is taken over all stopping times $\tau\in[0,T]$ with respect to $(\cF_t)_{t\in[0,T]}$. In what follows we will refer to \eqref{value1} as the \emph{PPSO problem}.

The value of the surrender option embedded in the contract (usually referred to as Early Exercise Premium in the mathematical finance literature) can be obtained as
\begin{align}\label{eq:pricing}
V^{\text{opt}}:=V_0-V^E_0,
\end{align}
where $V^E_0$ is the arbitrage-free price of the contract without the possibility of an early surrender, that is
\begin{align}\label{VE}
V^E_0 =\EQ\Big[e^{-r\, (T\wedge{\tau^\dagger})} g(A_{T\wedge\tau^\dagger}, R_{T\wedge\tau^\dagger})\Big].
\end{align}

It is worth noticing that, in practice, the use of surrender options may be disincentivised by insurance companies, who agree to pay out only a fraction of the policy reserve in case of early surrender. In our case that would correspond to take $\lambda g(A_\tau,R_\tau)$ in \eqref{value1} on the event $\{\tau<\tau^\dagger\wedge T\}$, with $\lambda\in[0,1)$. Here we focus on the model set out in \eqref{value1}, i.e., $\lambda=1$, which is consistent with the existing literature (see, e.g., \cite{GROSEN}, \cite{SIU}) and provides an upper bound for the prices of contracts with $\lambda\in[0,1)$. As we will see below, this model is also the source of interesting mathematical findings from the point of view of optimal stopping theory.


\subsection{Dimension reduction and bonus distribution rate}\label{sec:dimensionred}

As noticed in \cite{CHUKWOK} and \cite{SIU}, the PPSO problem can be made more tractable by considering the bonus distribution rate \eqref{eq:BDR} (i.e., the logarithm of the ratio $A/R$) as the observable process in the optimal stopping formulation of \eqref{value1}. Indeed, set $X:=(X_t)_{t\in[0,T]}$ with
\begin{align}\label{defX}
X_{t} := \ln \left(\frac{A_{t}}{R_{t}}\right)\qquad\textrm{for $t\in[0,T]$, $\Q$-a.s.}
\end{align}
Then, by \eqref{dynA} and \eqref{dynP}, one gets
\begin{align}\label{dynX}
\ud X_{t}= \Big(r-r^{G} -\tfrac{1}{2}\sigma^{2}-\big[\delta( X_{t}-\beta)-r^{G}\big]^{+} \Big)\ud t + \sigma \ud \WW_{t},
\end{align}
with initial condition
\begin{align}\label{xalpha}
X_0=x_\alpha:=\ln(1/\alpha)>0.
\end{align}
In terms of $X$ the insolvency time becomes
\begin{equation}\label{taumorto1}
\tau^{\dagger} = \inf\{t\ge 0\,:\, X_{t} \le 0 \}.
\end{equation}
Next we write the intrinsic value of the policy $g(A_t,R_t)$ (see \eqref{dynP}) in terms of $X_t$. For $x\in\R_+:=[0,\infty)$ define the gain function
\begin{align}\label{defh}
h(x) :=e^{-x} +\gamma\big[\alpha-e^{-x}\big]^{+},
\end{align}
and notice that
\begin{align}\label{eq:ph}
x\mapsto h(x)\:\:\text{is convex, striclty decreasing and}\:\:|h(x)-h(y)|\le |x-y|,
\end{align}
since $\gamma\in(0,1)$. Then
\begin{align}\label{gX}
g(A_t,R_t)=A_t\Big(e^{-X_t}+\gamma\big[\alpha-e^{-X_t}\big]^+\Big)=A_t\,h(X_t),\qquad\text{$t\in[0,T]$, $\Q$-a.s.}
\end{align}
From the expression above we notice that for $X_t>x_\alpha$ the participation bonus in the intrinsic value of the policy is strictly positive. So we can think of $x_\alpha$ as the activation threshold for the participation bonus.

Now the key to the dimension reduction is a change of measure. Define the martingale process $M:=(M_t)_{t\in[0,T]}$ by
\begin{align}\label{gXb}
M_t:=e^{\sigma \WW_t-\frac{\sigma^2}{2}t}=e^{-rt}A_t/a_0,
\end{align}
and the probability measure $\P$ equivalent to $\Q$ on $\cF_T$ given by $\ud\P=M_T\, \ud\Q$. By Girsanov theorem the process $W:=(W_t)_{t\in[0,T]}$ with
\begin{align}\label{Wtilde}
W_{t} := \WW_{t} - \sigma t,
\end{align}
is a $\P$-Brownian motion. Then, under the new measure $\P$, the dynamics of $X$ reads
\begin{equation}\label{X1}
\begin{cases}
\ud X_{t}= \pi(X_{t})\ud t + \sigma \ud W_{t},  \\
X_{0} = x_\alpha>0,
\end{cases}
\end{equation}
where
\begin{align}\label{pi}
\pi(x):=r-r^{G} +\frac{1}{2}\sigma^{2}-\big[\delta( x-\beta)-r^{G}\big]^{+}.
\end{align}
For future reference, it is worth defining 
\begin{align}\label{eq:xbar}
\bar x_0:=\beta+\frac{r}{\delta}\quad\text{and}\quad x_G:=\beta+\frac{r^G}{\delta}.
\end{align}
Then $x_G<\bar x_0$ since $r^G<r$. When the BDR process exceeds $x_G$ (i.e., $X_t>x_G$) the policyholder receives the bonus interest rate on the reserve, above the minimum rate guaranteed $r^G$ (see \eqref{scheme}). Again by \eqref{scheme} the interest rate paid on the reserve is higher than the risk-free rate $r$ when the BDR process exceeds $\bar x_0$ (i.e., $X_t>\bar x_0$).

Using \eqref{gX}, \eqref{gXb} and the optional sampling theorem, for any stopping time $\tau\in[0,T]$ one easily obtains
\begin{align}\label{tiz10}
\EQ &\Big[e^{-r\, (\tau\wedge{\tau^\dagger})} g(A_{\tau\wedge\tau^\dagger}, R_{\tau\wedge\tau^\dagger})\Big]=a_0\,\EQ\Big[M_{\tau\wedge\tau^\dagger}\,h\big(X_{\tau\wedge\tau^\dagger}\big)\Big]=
a_0\,\E\Big[h\big(X_{\tau\wedge\tau^\dagger}\big)\Big],
\end{align}
with $\E[\,\cdot\,]$ denoting the $\P$-expectation. Hence, \eqref{value1} may be rewritten as $V_0=a_0v_0$, where
\begin{align}\label{tiz11}
v_0:=\sup_{0\le \tau\le T}\E\big[h\big(X_{\tau\wedge\tau^\dagger}\big)\big].
\end{align}

Life-insurance contracts often charge management fees to the holder. Methods developed in this paper can be used to deal with fees that are charged at a rate proportional to the value of the portfolio $A$ and to the policy reserve $R$. If we add proportional fees to our model, the arbitrage-free price of the PPSO changes from its expression in \eqref{value1} to
\begin{align}\label{value:fees}
V_0 = \sup_{0\leq\tau\leq T}\EQ\Big[e^{-r\, (\tau\wedge{\tau^\dagger})} g(A_{\tau\wedge\tau^\dagger}, R_{\tau\wedge\tau^\dagger})-\int_0^{\tau\wedge\tau^\dagger}e^{-rs}\Big(pA_s+qR_s\Big)\ud s\Big],
\end{align}
for some $p,q\ge 0$ (notice in particular that $p=p(\alpha)$ should depend on the fraction $\alpha$ of the portfolio backing the policy purchased by the policyholder). Also in this case we can perform the change of measure displayed above and arrive at $V_0=a_0v_0$, where now
\begin{align}\label{v0:fees}
v_0=\sup_{0\le \tau\le T}\E\Big[h\big(X_{\tau\wedge\tau^\dagger}\big)-\int_0^{\tau\wedge\tau^\dagger}\Big(p+q e^{-X_s}\Big)\ud s\Big].
\end{align}

It is important to emphasise that, given a participation level $\alpha>0$, the value $V_0$ of the PPSO is specific to such value of $\alpha$. Indeed both the initial value of the reserve $R_0$ and the participation bonus on the intrinsic value of the policy depend on $\alpha$. So we should think of the PPSO's value as $V_0=V_0(\alpha)$ (and equivalently $v_0=v_0(\alpha)$). In order to solve the problem, i.e.,  determine $V_0$ and the optimal exercise time of the surrender option, in the next section we will embed our problem in a Markovian setting.

\section{Summary of main results}\label{sec:mainres}

Thanks to the Markovian nature of the process $X$ the value $v_0$ only depends on the initial value of the process $X_0=x_\alpha$ (see \eqref{xalpha}) and on the maturity $T$ of the contract. However, in order to be able to characterise $v_0$ and the associated optimal stopping rule, we must embed our problem into a larger state-space by considering all possible initial values of the time-space dynamics $(t,X)$. For that we denote by $X^x$ the process $X$ starting at time $0$ from an arbitrary point $x\ge 0$ and evolving according to \eqref{X1}. Similarly, we denote by $X^{t,x}$ the process $X$ starting at time $t$ from $x\ge 0$ and evolving according to \eqref{X1}. For future reference, when notationally convenient we will use $\E_x[\,\cdot\,]:=\E[\,\cdot\,|X_0=x]$. Since $X$ is time-homogeneous, it holds
\[
\mathsf{Law}\big((s,X^{t,x}_{s})_{s\ge t}\big)=\mathsf{Law}\big((t+s,X^x_s)_{s\ge 0}\big).
\]
Then we can identify the dynamics $(s,X^{t,x}_s)_{s\in[t,T]}$ and $(t+s,X^x_s)_{s\in[0,T-t]}$, and use the latter in the problem formulation below. Thanks to time-homogeneity, we also have that $\tau^\dagger$ defined in \eqref{taumorto1}
is independent of time. Sometimes we use $\tau^\dagger(x)$ to emphasise that $\tau^\dagger$ depends on $X_0=x$.

From now on we will study the finite-time horizon optimal stopping problem given by
\begin{align}\label{valuev}
v(t,x):=\sup_{0\le\tau\le T-t}\E\Big[h(X^{x}_{\tau\wedge\tau^\dagger})\Big],\qquad (t,x)\in[0,T]\times\mathbb{R}_+,
\end{align}
which embeds the PPSO problem \eqref{tiz11}. It is clear that we can go back to our original problem in two steps: first $v_0=v(0,x_\alpha)$, and then $V_0=a_0v_0$.
In the presence of management fees, by the same procedure we obtain the analogue of \eqref{valuev}:
\begin{align}\label{v:fees}
v(t,x)=\sup_{0\le\tau\le T-t}\E\Big[h(X^{x}_{\tau\wedge\tau^\dagger})-\!\int_0^{\tau\wedge\tau^\dagger}\!\!\Big(p\!+\!q e^{-X^x_s}\Big)\ud s\Big],\quad (t,x)\in[0,T]\times\mathbb{R}_+.
\end{align}
The study of \eqref{valuev} and \eqref{v:fees} are equivalent from the methodological point of view (see Section \ref{sec:fees} for a detailed discussion) and therefore we focus on \eqref{valuev} in the interest of notational simplicity. We show in Section \ref{sec:fees} (Figure \ref{fig4}) how the addition of management fees affects the qualitative properties of the surrender policy.

\subsection{Theoretical results: value function and optimal exercise boundary}
In this section we provide the main theoretical results of the paper, whose proofs are given at the end of Section \ref{sec:freeb} and build upon technical results obtained in Sections \ref{sec:value} and \ref{sec:freeb} for the ease presentation.
Let $\cL$ be the second order differential operator associated to the diffusion \eqref{X1}, i.e.
\begin{align}\label{eq:L}
(\mathcal{L} f)(x):=\tfrac{\sigma^2 }{2} \partial_{xx}f (x)+\pi(x) \partial_xf(x),\quad\text{for any $f\in C^2(\R_+)$},
\end{align}
with $\partial_x$ and $\partial_{xx}$ denoting the first and second order partial derivatives with respect to $x$, respectively. We shall also denote the partial derivative with respect to time by $\partial_t$. As usual in optimal stopping theory, let us introduce
\begin{align}\label{setC0}
\mathcal{C}=&\, \big\{ (t,x)\in [0,T]\times \R_+ \ : \ v(t,x)>h(x) \big\},
\end{align}
and
\begin{align}\label{setS0}
\mathcal{S}=&\, \big\{ (t,x)\in [0,T]\times \R_+ \ : \ v(t,x)=h(x) \big\},
\end{align}
that are the so-called continuation and stopping regions, respectively. For future reference let $\partial\cC$ be the boundary of the set $\cC$ (notice that $\{T\}\times \R_+\subseteq\partial\cC$) and introduce the first entry time of $(t+s,X_s)$ into $\mathcal{S}$, i.e.
\begin{equation}
\label{entry}
\tau^*(t,x):= \inf \left \{ s \in [0,T-t] \ :  \ (t+s,X_s^x) \in \cS \right \}.
\end{equation}

On the value function $v$ of \eqref{valuev} we have the next result.
\begin{theorem}[The value function]\label{thm:v}
The function $v$ is non-negative, continuous and bounded by $1$ on the set $[0,T]\times\R_+$, with $v\ge h$. The mappings $t\mapsto v(t,x)$ and $x\mapsto v(t,x)$ are both non-increasing. Moreover, $v\in C^1\big([0,T)\times(0,\infty)\big)$, the second derivative $\partial_{xx}v$ exists and is continuous on the set $\overline\cC\cap\big([0,T)\times(0,\infty)\big)$ and $v$ solves (uniquely) the free boundary problem
\begin{align}\label{PDEu0}
\left\{
\begin{array}{ll}
\partial_t v+\cL v = 0,& \text{in $\cC$},\\
\partial_t v+\cL v \le 0,& \text{in $\cS$},\\
v\ge h,& \text{on $[0,T]\times\R_+$},\\
v=h,& \text{on $\partial\cC$.}
\end{array}
\right.
\end{align}
\end{theorem}

In the PDE literature, the above result is often presented in terms of a variational inequality. That is, $v$ is the unique solution, in the a.e.\ sense of the obstacle problem:
\[
\max\{\partial_t v+\cL v,h-v\}=0,\quad\text{on $[0,T)\times\R_+$\,,}
\]
with boundary conditions $v(T,x)=h(x)$ for $x\in\R_+$ and $v(t,0)=h(0)$ for $t\in[0,T)$. Uniqueness in Theorem \ref{thm:v} refers to the class of continuous functions $w$ such that $w\in C^1([0,T)\times(0,\infty))$ with $w_{xx}\in L^\infty_{\ell oc}([0,T)\times(0,\infty))$.

From continuity of $v$ we deduce that $\cC$ is open and $\cS$ is closed. Hence in particular $\partial\cC\subset\cS$. Moreover, standard optimal stopping results (see \cite[Cor.\ 2.9, Sec.\ 2]{PS}) guarantee that the entry time $\tau_*$ to the stopping set $\cS$ \eqref{entry} is optimal for $v(t,x)$.
It is then of interest to determine the geometry of $\cS$.

At first we notice that since $t\mapsto v(t,x)-h(x)$ is non-increasing, then we can define 
\begin{align}\label{eq:c}
c(x):=\inf\{t\in[0,T]: v(t,x)=h(x)\},\quad \text{for $x\in\R_+$ (recall $\R_+:=[0,\infty)$).}
\end{align}
This gives us a parametrisation of the stopping set as 
\begin{align}\label{eq:Sc}
\cS=\{(t,x)\in[0,T]\times\R_+: t\ge c(x)\}. 
\end{align}
In optimal stopping theory and its financial applications it is often preferable to describe the set $\cS$ in terms of time-dependent boundaries. So, in our analysis in Sections \ref{sec:value}--\ref{sec:freeb} we use the boundary $c(\,\cdot\,)$ as a useful technical tool but we are also able to prove that it can be inverted locally and we present our results here in terms of time-dependent boundaries $b_1$, $b_2$ and $b_3$.

It turns out that there are two different shapes of $\cS$ depending on the model parameters. Recalling $\bar x_0=\beta+\frac{r}{\delta}$ and $x_\alpha=\ln(1/\alpha)$ we will address separately the cases $x_\alpha<\bar x_0$ and $x_\alpha\ge \bar x_0$.
Our second main result is summarised below, where we denote $f(t-)$ the left limit of a function $f$ at a point $t$ and we adopt the convention $[t,t)=\varnothing$ for any $t\in\R$.

\begin{theorem}[The optimal boundary]\label{thm:boundary}
The following holds:
\begin{itemize}
\item[(a)] If $x_\alpha\ge\bar x_0$, there exists a function $b_1:[0,T)\to [0,\bar x_0]$ and a constant $t_0\in[0,T)$ such that $b_1(t)=0$ for $t\in[0,t_0)$, $b_1$ is strictly increasing and continuous on $[t_0,T)$ with $b_1(T-)=\bar x_0$, and the stopping region is of the form
\[
\cS=\big\{(t,x)\in[0,T)\times\R_+: x\le b_1(t)\big\}\cup\big(\{T\}\times\R_+\big).
\]
Thus the optimal stopping time $\tau_*$ reads
\[
\tau^*=\inf\{t\in[0,T)\,:\,X_t\le b_1(t)\}\wedge T.
\]

\item[(b)] If $x_\alpha<\bar x_0$, there exist constants $t_0\in[0,T)$ and $\hat c\in[0,T]$, and functions $b_1:[0,T)\to [0, x_\alpha]$ and $b_2,b_3:[\hat c,T)\to [x_\alpha,\bar x_0]$, such that:
\begin{itemize}
\item[  i)] $b_1(t)=0$ for $t\in[0,t_0)$ and $b_1$ is strictly increasing and continuous on $[t_0,T)$;
\item[ ii)] $b_2$ is strictly decreasing and continuous while $b_3$ is strictly increasing and continuous on $[\hat c,T)$;
\item[iii)] $b_1(t)\le b_2(t)\le b_3(t)$ for $t\in[\hat c,T)$ with $b_2(\hat c)=b_3(\hat c)$ if $\hat c>0$, and $\hat c=0$ if $b_2(\hat c)<b_3(\hat c)$;
\item[ iv)] $b_1(T-)=b_2(T-)=x_\alpha$ and $b_3(T-)=\bar x_0$;
\item[  v)] the stopping region is of the form
\[
\cS=\big\{(t,x)\in[0,T)\times\R_+: x\le b_1(t)\:\:\text{or}\:\:x\in[b_2(t),b_3(t)]\big\}\cup\big(\{T\}\cup\R_+\big),
\]
thus the optimal stopping time $\tau_*$ reads
\[
\tau^*=\inf\{t\in[0,T)\,:\,X_t\le b_1(t)\:\:\text{or}\:\:X_t\in[b_2(t),b_3(t)]\}\wedge T.
\]
\end{itemize}
\end{itemize}
\end{theorem}

\begin{remark}\label{rem:xa}
A close inspection of the theorem above shows that $[0,T)\times\{x_\alpha\}\subset\cC$ in all cases (see Proposition \ref{prop1} for the proof).
\end{remark}

As anticipated in the Introduction, in case (b) we have two `stop-loss' boundaries (i.e., $b_1$ and $b_3$), which trigger the surrender option when the BDR process crosses them downwards, and a `too-good-to-persist' boundary (i.e., $b_2$), which triggers the surrender option when the BDR process crosses it upwards. In case (a) instead we only observe a single stop-loss boundary. These results will be interpreted in Section \ref{sec:num} below.

At the technical level, the strict monotonicity and continuity of the time-dependent boundaries in (a) and (b) of the theorem above are derived by analogous properties for the $x$-dependent boundary from \eqref{eq:c}. In particular, in case (a) we will prove that the function $x\mapsto c(x)$ is continuous on $(0,\infty)$, there exists $x_1\in[0,\bar x_0)$ such that $c(x)=0$ for $x\in[0,x_1)$ and it is strictly increasing on $(x_1,\bar x_0)$. So we have 
\[
b_1(t):=\inf\{x\in[0,\bar x_0): c(x)>t\},\qquad t\in[0,T).
\]
In case (b) instead the geometry is more involved. We will prove that there exist $x_1\in[0,x_\alpha)$ and $x_\alpha<x_2\le x_3<\bar x_0$ such that:
\begin{itemize} 
\item[(i)] $x\mapsto c(x)$ is continuous on $(0,\infty)$, $c(x)=0$ for $x\in[0,x_1)$ and it is strictly increasing on $(x_1,x_\alpha)$; 
\item[(ii)] the set of minimisers of $x\mapsto c(x)$ in $(x_\alpha,\bar x_0)$ is the closed interval $[x_2,x_3]\subset (x_\alpha,\bar x_0)$  where $c(\,\cdot\,)$ takes the value $\hat c$\,; if $x_2<x_3$ then $\hat c=0$;
\item[(iii)] $c(\,\cdot\,)$ is  strictly decreasing on $(x_\alpha,x_2)$ and strictly increasing on $(x_3,\bar x_0)$. 
\end{itemize}
Then, the boundaries $b_1$, $b_2$ and $b_3$ are obtained as
\begin{align*}
&b_1(t):=\inf\{x\in[0,x_\alpha):c(x)>t\},\qquad\:\:\:\: t\in[0,T),\\
&b_2(t):=\sup\{x\in(x_\alpha,x_2):c(x)>t\},\qquad t\in[\hat c,T),\\
&b_3(t):=\inf\{x\in(x_3,\bar x_0):c(x)>t\},\qquad\:\: t\in[\hat c,T).
\end{align*}
See Figures \ref{fig2}, \ref{fig3} and \ref{fig4} for various illustrations with both $\hat c=0$ and $\hat c>0$ and with $t_0=0$ and $t_0>0$ (notice that $t_0=\lim_{x\downarrow 0}c(x)$).

\subsection{Some technical remarks}
The choice to work with the boundary $x\mapsto c(x)$ in our theoretical analysis is dictated by the fact that {\em a priori} it seems too difficult to establish existence of the three boundaries $b_1$, $b_2$ and $b_3$. Indeed this would normally require to prove piecewise monotonicity of the map $x\mapsto v(t,x)-h(x)$ and/or convexity of the map $x\mapsto v(t,x)$, plus developing arguments that guarantee non-emptyness of the set $\{x\in\R_+:v(t,x)=h(x)\}$ depending on the choice of $t\in[0,T)$. Neither of these tasks follows by standard arguments because of the lack of an explicit solution for the SDE \eqref{X1} and due to the absorption at $x=0$ for the dynamics of $X$.

The probabilistic proof of the strict monotonicity of the optimal exercise boundaries in Theorem \ref{thm:boundary} is an interesting technical result in its own right and so far it was missing from the optimal stopping literature. In the PDE literature strict monotonicity of free boundaries (and even their smoothness) are well-known results. Classical references are the monographs \cite{LSU} and \cite{Fr} for a general treatment, while for parabolic problems with one spatial dimension (i.e., closer to our set-up) one can also refer to \cite{Cannon} and \cite{Fr75}. The techniques developed in those seminal contributions were then employed and tailored for optimal stopping problems in mathematical finance as for example in \cite{BX09,CC07} for American option pricing and \cite{DY09} for optimal investment with transaction costs. 

For the smoothness of the boundary (understood as its continuous differentiability or higher), PDE arguments often require smoothness of the obstacle (i.e., the option's payoff) and in all cases continuous differentiability of the coefficients of the SDE underlying the stochastic optimisation. When the obstacle is not smooth (as in the American put problem) one often takes advantage of the explicit transition density of the underlying stochastic process (typically a geometric Brownian motion). 
In our case we have neither a smooth payoff (see \eqref{defh}) nor continuously differentiable coefficients (see \eqref{pi}). Moreover, the transition density of our process $X$ is not known. So we cannot apply classical results from the PDE literature to derive continuous differentiability of the optimal boundary and we set this question aside. 
 
For the strict monotonicity of free boundaries the PDE literature relies upon an application of Hopf's lemma and an argument by contradiction. 
Our probabilistic proof complements those PDE techniques by employing methods more familiar to probabilists working on optimal stopping.

The shape of the stopping region in (b) of Theorem \ref{thm:boundary} is somewhat remarkable and it was never observed in the context of participating policies with surrender options. Not only the stopping region is disconnected, but when $\hat c>0$ there is a point in the stopping region at which one of the stop-loss boundaries meets the too-good-to-persist boundary (see Figure \ref{fig1}). Similar geometries of optimal stopping regions in the time-space plane have been observed numerically (see, e.g., \cite[Fig.4]{EV20}) but a complete theoretical analysis is not usually available. An instance of such study, revealing a similar geometry in a finite horizon optimal stopping problem, is \cite{DuTPe07}. However, the problem studied in \cite{DuTPe07} concerns the optimal prediction of the maximum of a Brownian motion with drift, whereas the one studied in \cite{EV20} concerns stopping of a partially observable Brownian bridge. Hence, the similarities with the stopping rule in our set-up appear to be a mere coincidence. 

Before presenting the full proofs of Theorems \ref{thm:v} and \ref{thm:boundary}, in the next section we discuss in detail the financial interpretation of our results with the aid of extensive numerical tests. The complete theoretical analysis that leads to Theorems \ref{thm:v} and \ref{thm:boundary} is performed in Sections \ref{sec:value} and \ref{sec:freeb} for the interested reader.


\subsection{Numerical results and financial interpretation}\label{sec:num}

In order to investigate the shape of the continuation and stopping regions we implement a binomial-tree algorithm  based on the diffusion approximation scheme proposed in \cite{NR}. We take a partition of $[0,T]$ with $N+1$ equally spaced time points. At each node in the tree we associate a value of the underlying process $X$ and of the corresponding time: that is, in the $(n,j)$-node we have the couple $(n,x^j_n)$ for $j=0,1,\ldots n$. At the subsequent time-step the process can move to one of the two nodes $(n+1,x^j_n\pm\sigma\sqrt{\Delta})$ with $\Delta:=T/N$, so that the tree is recombining. If $x^j_n>0$ the probability $p^j_n$ of moving upwards from the $n$-th node is calculated as in \cite{NR} as $p^j_n=0\vee[1\wedge(\tfrac{1}{2}+\sqrt{\Delta}\cdot\pi(x^j_n)/2\sigma)]$. If instead $x^j_n\le 0$ the process can only  move to $(n+1,0)$ with probability one. We compute the numerical approximation of the value function $\tilde v_n(x^j_n)$ of the PPSO with the usual backward recursion, starting from $\tilde v_N(x^j_N)=h(x^j_N)$ for $x^j_N\ge 0$. For any $n<N$, if $x^j_n\le 0$, then $\tilde v_n(x^j_n)=h(0)=1$; if instead $x^j_n> 0$, then $\tilde v_n(x^j_n)=\max\{h(x^j_n),\E[\tilde v_{n+1}(X_{n+1})|X_n=x^j_n]\}$. Since the binomial-tree has recombining nodes, the evaluation of the continuation value $\E[\tilde v_{n+1}(X_{n+1})|X_n=x^j_n]$ reduces to the average of the payoff at the next two nodes.

\begin{remark}
Notice that the regularity we have obtained for the value function $v$ allows us, in principle, to obtain an integral equation for the optimal boundary (see \cite{PS} for some examples). However, as the explicit form of the transition density of the process $X$ is not known, solving such integral equation numerically would not be possible. This motivates our use of binomial-trees.
\end{remark}

Unless otherwise specified, in the rest of the section we set the following values for the parameters (time is expressed in years while $r$, $r^G$ and $\sigma$ are annual rates) 
\begin{equation}\label{eq:param}
T=10,\, r=1.5\%,\, \sigma=18\%,\, r^G=1\%,\, \delta=0.1,\, \gamma=0.4,\, \beta=3,\, \alpha=0.1.
\end{equation}
For such parameter values $x_\alpha=2.3$, $\bar{x}_0=3.15$ and $x_G=3.1$ (see \eqref{xalpha} and \eqref{eq:xbar} respectively), therefore we are in the setting of $x_\alpha<\bar x_0$ (see (b) in Theorem \ref{thm:boundary}).

\subsubsection{Financial interpretation of the surrender region}

We recall that the structural properties of the intrinsic value of the policy (i.e., the function $h$ in \eqref{tiz11}) and the initial value of the reserve $R_0$ are determined by the choice of $\alpha$. From the financial perspective the PPSO is priced for each fixed value of the parameter $\alpha$, that is $V_0(\alpha)=a_0v_0(\alpha)$ (recall \eqref{value1} and \eqref{tiz11}), and the policyholder's initial BDR process at time zero is $X_0=x_\alpha$. 
Figure \ref{fig1} shows the optimal surrender region $\cS$ and the boundary $c(\cdot)$ as in \eqref{eq:Sc} on the $(x,t)$ plane.

\subsubsection*{Optimal surrender and default} Notice that in Figure \ref{fig1} the solvency requirement is always fulfilled if the policyholder exercises the SO optimally (i.e., $\tau^*<\tau^\dagger$). This corresponds to $t_0=0$ in (i) of Theorem \ref{thm:boundary}-(b) and, in particular, $x_1:=b_1(0)>0$ (compare also with $x_1$ in Proposition \ref{prop:monot-c2}-(i)). Hence the optimal termination of the contract can only occur due to surrender or at maturity. A different situation appears in Figure \ref{fig2}-(i), where $t_0=c(0+)>0$ and early termination due to solvency requirements may occur if the dynamics of $X$ hits zero prior to time $t_0$.
\begin{figure}[t!]
	\centering
	\begin{center}
		\hspace*{-0.7cm}
		\includegraphics[scale=0.4]{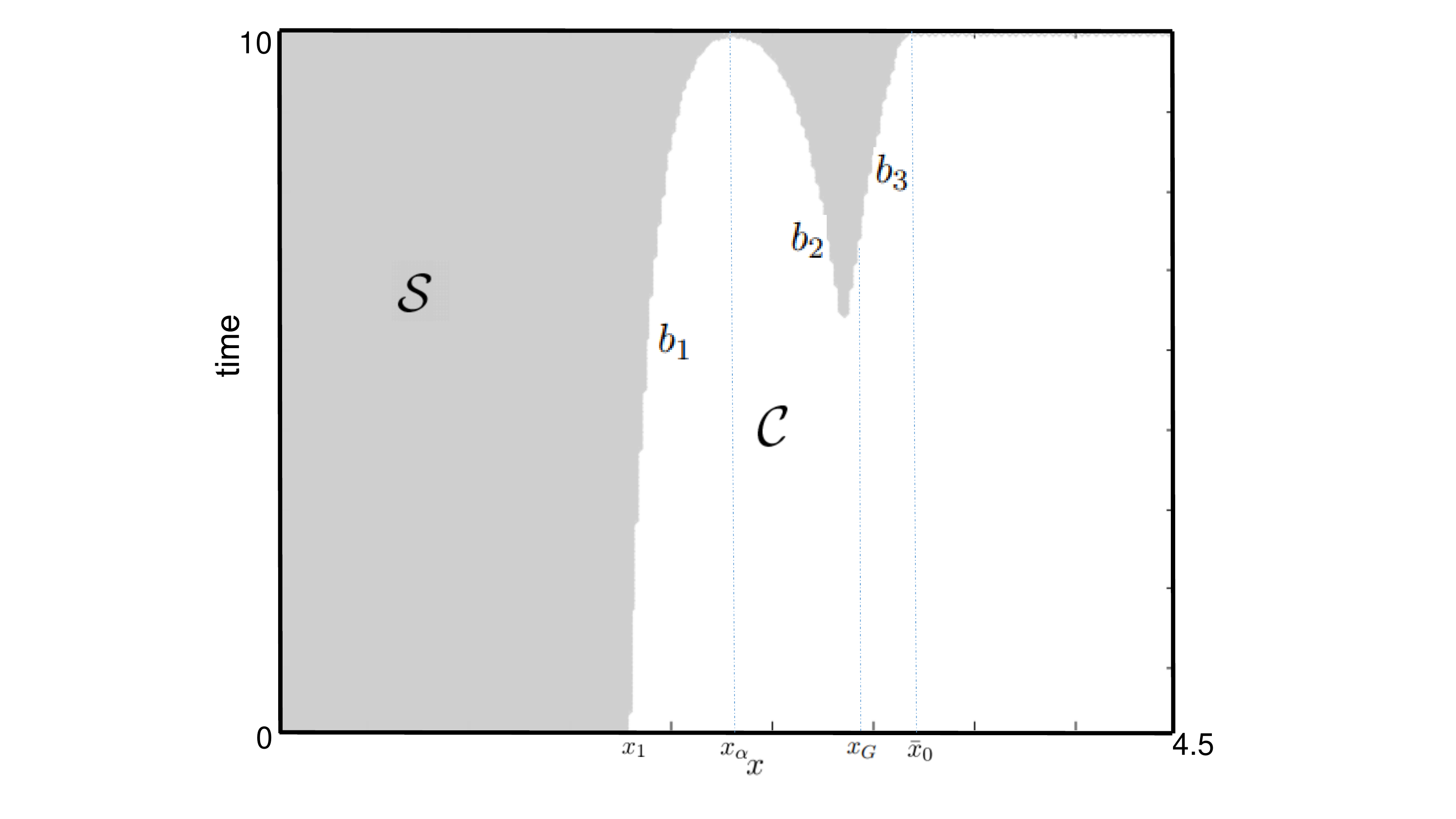}
		\vspace{-0.5cm}
	\end{center}
	\caption{The optimal surrender regions and boundaries in the case $x_1=b_1(0)>0$ and $b_2(\hat c)=b_3(\hat c)$ with $\hat c>0$   
(see Theorem \ref{thm:boundary}-(b)).}
	\label{fig1}
\end{figure}

\subsubsection*{Minimum rate guarantee, bonus rate and a stop-loss boundary} There is a natural interpretation for the shape of the surrender region for $X$ close to zero and for $X\ge \bar x_0$. On the one hand, when $X$ is close to $0$, $c(A,R) =r^G$ and $g(A,R) =R$ (see \eqref{scheme} and \eqref{gdef} respectively); thus, the policy reserve grows at rate $r^G$ which is lower than the discount rate $r$ used in \eqref{value1}. So the policyholder has an incentive to surrender in order to avoid an erosion of the present value of the reserve (which is due to the gap $r^G-r<0$). Hence it is natural to interpret the boundary $b_1$ as a stop-loss boundary. On the other hand, for values of $X$ larger than $\bar{x}_0$, the policy reserve grows at a rate greater than $r$, due to the bonus mechanism in \eqref{scheme}. In this case the policyholder has no incentive to surrender the contract and the stopping region disappears.

\subsubsection*{Stop-loss and too-good-to-persist boundaries for $X\in(x_\alpha,\bar x_0)$}
The peculiar shape of the surrender region between $x_\alpha$ and $\bar x_0$ can be explained as follows. The value $x_\alpha$ is the critical value at which the participation bonus in the intrinsic value of the policy becomes active (see \eqref{gdef}). Then, if $X=x_\alpha$, the investor delays the surrender with a view to possibly receiving the bonus. Moreover, in a neighbourhood of $x_\alpha$ the drift in the dynamics \eqref{X1} is positive (hence pulling the BDR process towards the bonus), so that the policyholder has an incentive to wait also if $X<x_\alpha$ but not too small. When $X>x_\alpha$ the participation bonus in the policy's intrinsic value is active and can be collected by the policyholder upon immediate surrender. This gives origin to the too-good-to-persist boundary $b_2$. When $X>x_G$ the bonus on the policy reserve's growth rate is also active (see \eqref{scheme}) and creates an incentive to wait. If the BDR process is close to $\bar x_0$ surrendering is not appealing. In fact the policyholder stays in the contract hoping that $X$ will exceed $\bar x_0$ and the reserve will grow at a higher rate than the risk-free rate. However, if $X$ decreases, the participation bonus in the intrinsic value of the policy lessens. At the same time, for $X\le \bar x_0$ the growth rate of the reserve is still smaller than the risk free rate. So, as the maturity approaches, the combined effect of these two mechanisms creates the stop-loss boundary $b_3$. 

\subsubsection{Sensitivity analysis} Here we discuss the impact of various parameters on the shape of the surrender region and on the value of both the policy and the surrender option. In what follows both the parametrisations of the boundary $\partial\cC$ in terms of $c$ and in terms of $b_1$, $b_2$, $b_3$ will be used.

\subsubsection*{The $\alpha$ fraction of the initial portfolio}   Figure \ref{fig2} shows possible shapes of the optimal surrender boundary $c(\cdot)$ that complement the one presented in Figure \ref{fig1}. The plots are obtained for several values of the parameter $\alpha$, or equivalently $x_\alpha$ (see \eqref{xalpha}). The possible presence of a portion of the continuation region below the local minimum of $c(\,\cdot\,)$ (as in Figure \ref{fig1}) and the value $\hat c$ of the minimum itself depend on several factors, including the value of $\alpha$. Large values of $\alpha$ push the initial BDR $X_0=x_\alpha$ towards zero so that the chances of benefiting from the bonus on the policy reserve's growth rate \eqref{scheme} are slim and the policyholder will prioritise the participation bonus in the intrinsic value of the policy \eqref{gdef}. That widens the area above the local minimum of $c(\,\cdot\,)$ until the continuation region \eqref{setC0} becomes completely disconnected (see Figure \ref{fig2}-(i)). It should be emphasised that since $X_0=x_\alpha$, a completely disconnected surrender region means that the policyholder will surrender the contract as soon as the BDR process leaves the continuation region between the lower stop-loss boundary $b_1$ and the too-good-to-persist boundary $b_2$. As $\alpha$ increases this mechanism is reversed: the activation thresholds $x_\alpha$ of the participating bonus and $x_G$ of the bonus on the reserve's growth rate become closer. Then the portion of stopping region above the local minimum of $c(\,\cdot\,)$ shrinks as $x_\alpha$ approaches $\bar x_0$ from the left (see Figure  \ref{fig2}-(ii) and Figure \ref{fig1} where $x_\alpha=2.3$). In the limit we arrive to the case of $x_\alpha\ge \bar x_0$ (cf. \eqref{eq:case2}), which is also illustrated in Figure  \ref{fig2}-(iii).
There the situation is less involved because the participating bonus kicks in after the process $X$ has already exceeded $\bar x_0$, so that the reserve is already growing at a rate higher than the discount rate. In that case the policyholder's waiting strategy is aimed at collecting both a large reserve and the participating bonus. The exercise of the SO in this setting is only optimal when $X$ is sufficiently small and it is purely triggered by the stop-loss mechanism due to discounting.

\begin{figure}[t!]
	\centering
	\begin{center}
		\hspace*{-0.7cm}
		\includegraphics[scale=0.4]{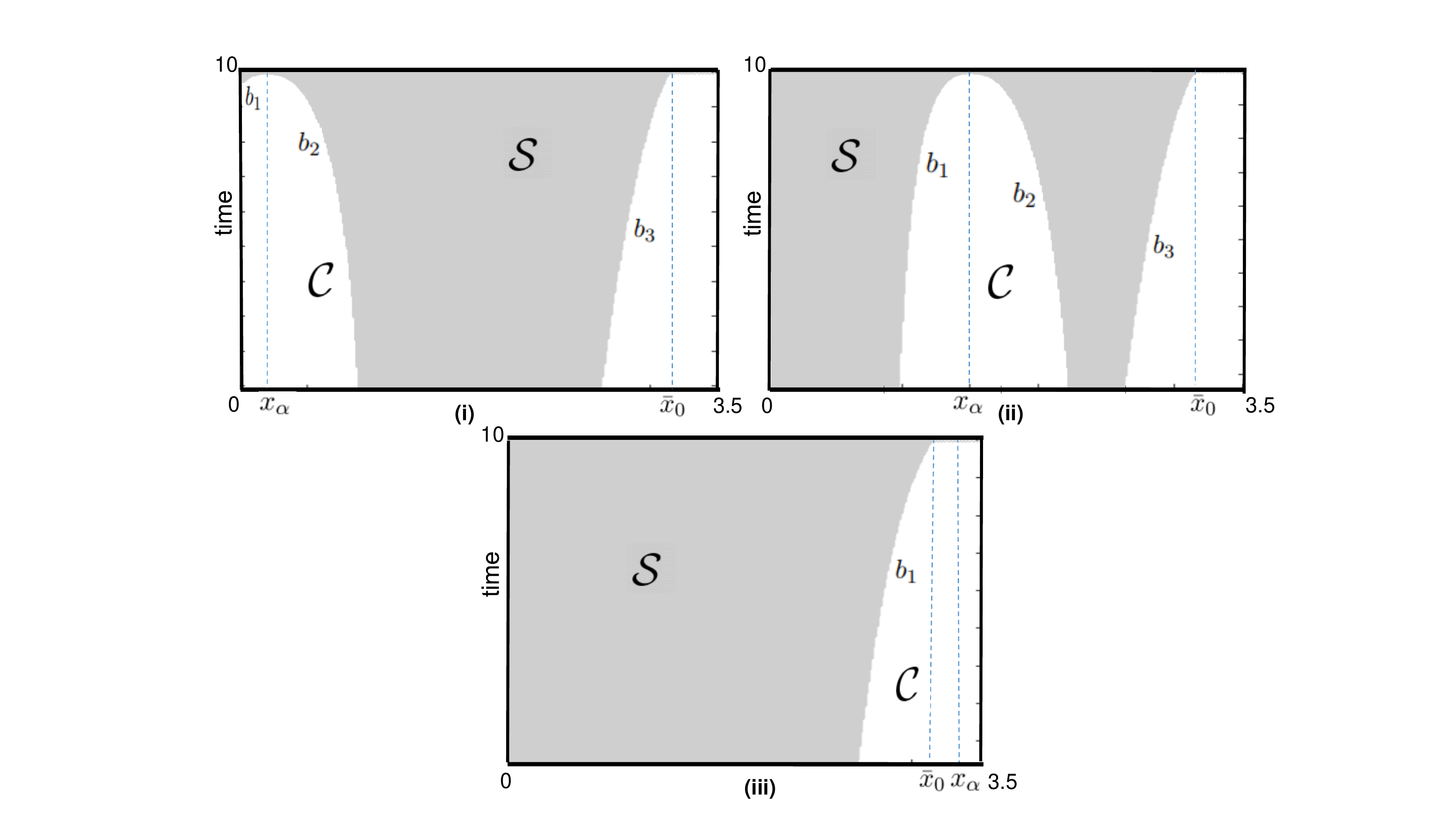}
		\vspace{-0.5cm}
	\end{center}
\caption{The optimal surrender boundary when varying $x_\alpha$. Here $\bar x_0=3.15$ and \textbf{(i)} $ x_\alpha=0.2$, \textbf{(ii)} $x_\alpha=1.5$, \textbf{(iii)} $x_\alpha=3.3$.}
	\label{fig2}
\end{figure}

\subsubsection*{Participation coefficient and minimum rate guarantee}
In Figure \ref{fig3} we study the sensitivity of the optimal surrender boundary $c(\cdot)$ with respect to $\gamma$ (left plot) and to $r^G$ (right plot). We remark that $\gamma$ only affects the intrinsic value of the policy (see \eqref{gdef}) but it does not affect either the reference portfolio $A$ nor the policy reserve $R$. As $\gamma$ increases, the participating bonus increases and counters the effect of discounting. So the policyholder is inclined to stay in the contract longer to see if the the bonus mechanism on the reserve will also be activated. Likewise, as $r^G$ increases the policyholder has progressively more benefits from staying in the contract, then the SO becomes less appealing and the area between the boundaries $b_2$ and $b_3$ shrinks. In particular, as $r^G$ approaches the risk free rate $r$ the stop-loss boundary $b_3$ and the too-good-to-persist boundary $b_2$ do not disappear but become less extended (theoretically this is expected because the interval $(x_\alpha,\bar x_0)$ does not depend on $r^G$ and it is shown in Lemma \ref{prop:strip} that $\cS\cap\big([0,T)\times(x_\alpha,\bar x_0)\big)\neq\varnothing$). In both situations the incentive to surrender the contract decreases and the continuation region expands. As a result the optimal boundary $c(\cdot)$ is pushed upwards in our plots. 

\subsubsection*{Values of the policy and of the surrender option} 
We conclude the section by analysing how the bonus distribution mechanism and the minimum interest rate guarantee in the policy reserve impact on the value of the policy and on the value of the embedded SO. The value of the SO is obtained, as in \eqref{eq:pricing}, by comparing the value of the PPSO to the value of its European counterpart (i.e., with no SO).
\begin{figure}[ht!]
	\centering
	\begin{center}
		\hspace*{-2.4cm}
		\includegraphics[scale=0.55]{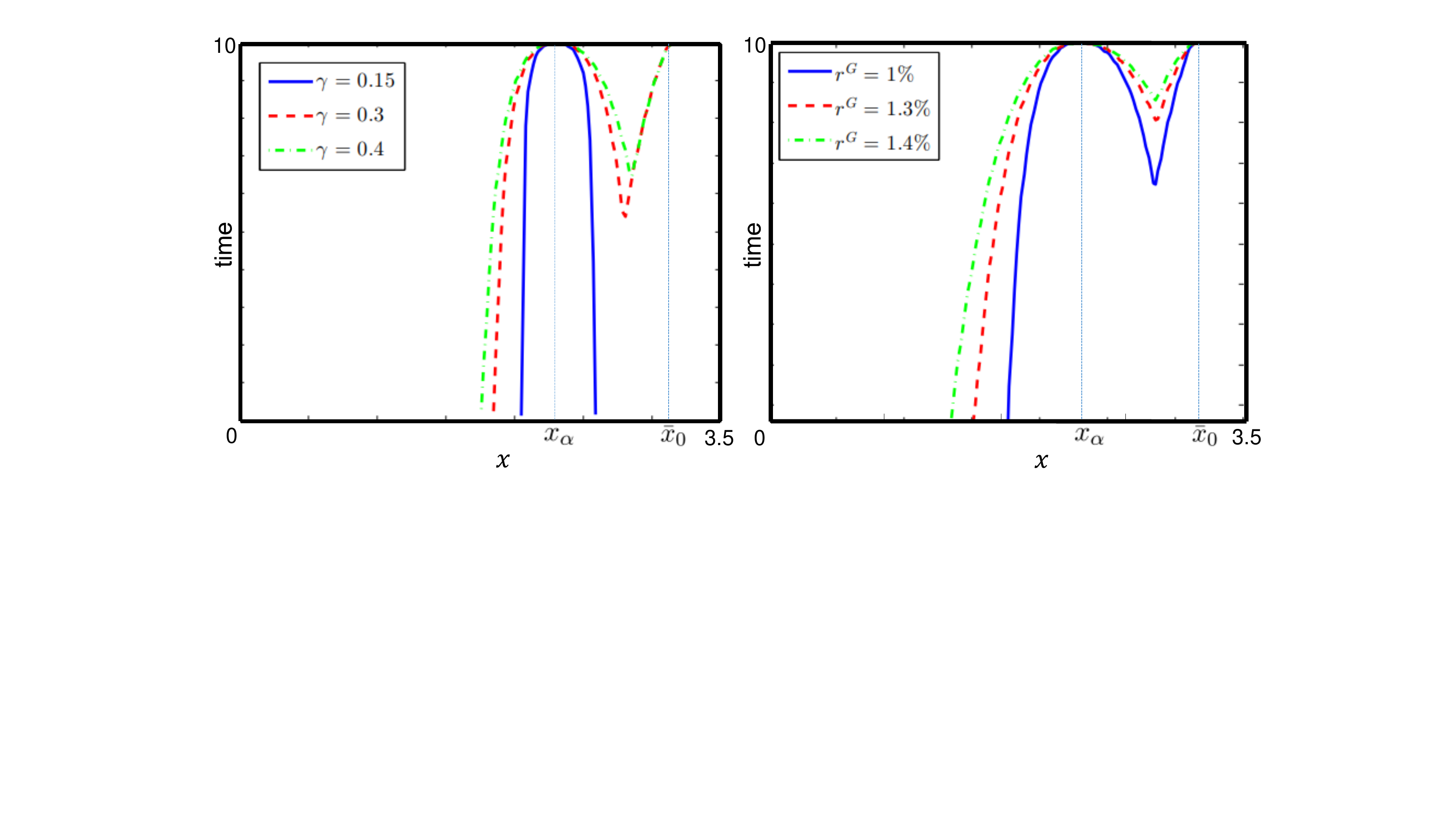}
		\vspace{-4.5cm}
	\end{center}
	\caption{ Sensitivity of the optimal surrender boundary $c(\cdot)$ with respect to $\gamma$ (left plot) and to $r^G$ (right plot). For $\gamma=0.15$ (left plot) we omit the rightmost portion of the boundary which is never reached when $X_0=x_\alpha$ and the policyholder stops optimally.}
		\label{fig3}
\end{figure}
Now we fix the initial portfolio value $A_0=1,000$ so that $R_0=100$.
We collect in Table \ref{tab1} the value $V_0$ of the PPSO (see \eqref{value1}), the value $V_0^E$ of the contract without SO (see \eqref{VE}) and the value $V^{\text{opt}}$ of the SO.
As in \cite{GROSEN}, we consider the following three scenarios depending on the level of participation in the returns generated by the reference portfolio: \textit{low} ($\delta=0.1$ and $\beta=3.4$), \textit{medium} ($\delta=0.25$ and $\beta=2.7$), \textit{high} ($\delta=0.6$ and $\beta=2$). The lower the value of $\delta$, the less the policyholder participates in the reserve via \eqref{scheme}. Moreover, the higher the value of the target buffer ratio $\beta$, the smaller is the surplus that the policyholder receives.
The value of $V_0^E$ is evaluated by using the same binomial-tree method described above, without the complication of the optimisation which is required at each node in the PPSO. As expected, $V_0$ is always greater than $V_0^E$. Their difference gives the option value $V^{\text{opt}}$.
\begin{table}[ht!]
	\centering
	\begin{tabular}{ |p{2.5cm} p{2cm}|p{2cm}|p{2cm}|p{2cm}|  }
		\hline
		Spread & Scenario &$V_0$&$V_0^E$&$V^{\text{opt}}$\\
		\hline
		$r-r^G=0.5\%$   & low    &100.7&   99.44 & 1.26\\
		& medium    &104.16&   103.43 & 0.73\\
		& high   &160.93&   160.41 & 0.52
		\vspace{0.5cm}\\
		$r-r^G=0.8\%$   & low    &100.27&   94.92 & 5.35\\
		& medium    &102.17&   99.29 & 2.88\\
		& high    &158.14&   156.47 & 1.67
		\vspace{0.5cm}\\
		$r-r^G=1.5\%$   & low    &100.14&   88.98 & 11.16\\
		& medium    &100.64&   93.93 & 6.71\\
		& high    &154.81&   151.38 & 3.43 \\
		\hline
	\end{tabular}
	\vspace{0.5cm}
	\caption{The values $V_0$ of the PPSO, $V_0^E$ of the contract without surrender option and $V^{\text{opt}}$ of the surrender option.}
	\label{tab1}
\end{table}

In the {\em low} scenario the value of the European contract $V_0^E$ is below par for all values of the spread $r-r^G$, i.e. $V^E_0<R_0$. If the spread is relatively large, i.e. $0.8\%$ or $1.5\%$, the European contract trades below par also in the \textit{medium} scenario. When the minimum interest rate guaranteed $r^G$ is much smaller than the market risk-free rate and the level of participation in the returns is also relatively small, the policy is not financially appealing to an investor if compared, for example, to bond investments. However, since $R_0>V^E_0$ (see Table \ref{tab1}) an investor who purchases the European contract incurs an initial outlay which is smaller than the initial amount credited to the reserve. This makes the policy potentially appealing as a form of secure savings.

The value $V_0$ of the PPSO is always at or above par, i.e. $V_0\ge R_0$, due to the American-type option embedded in the contract (at par the SO is immediately exercised).
Contract values $V_0$ and $V_0^E$ increase moving from \textit{low} towards \textit{high} scenario whereas the value of $V^{\text{opt}}$ decreases. This shows that the incentive to exercise the SO is reduced by higher participation of the investor in the returns. On the contrary, as $r-r^G$ increases the contract values decrease, whereas $V^{\text{opt}}$ increases. This is in line with the intuition that the higher the spread, the less the contract is profitable for the policyholder, hence creating a big incentive to exercise the SO.


\section{Properties of the value function}\label{sec:value}

In this section we collect some facts about the underlying stochastic process $X$, defined in \eqref{X1}, which will then be used to infer regularity of the value function \eqref{valuev}.

\subsection{Path properties of the underlying process}
First we observe that since the drift function $\pi(\,\cdot\,)$ is Lipschitz continuous and the diffusion coefficient is constant, there exists a modification $\widetilde X$ of $X$ such that the stochastic flow $(t,x)\mapsto \widetilde X^x_t(\omega)$ is continuous for a.e.~$\omega\in\Omega$ (see, e.g., \cite[Chapter V.7]{Protter}). As usual, throughout the paper we work with the continuous modification which we still denote by $X$ for simplicity.

\begin{lemma}\label{lemma-X}
For any $\P$-a.s.~finite stopping time $\tau\ge 0$ it holds
\begin{align}
\label{eq:lipX}&\big| X^{x}_{\tau} - X^{y}_{\tau} \big|\le \big| x-y \big| e^{\delta \tau}, \quad\text{$\P$-a.s.~for $x,y \in \R_+$},\\
\label{eq:bdX}& X^{y}_{\tau} - X^{x}_{\tau}\ge (y-x)(2-e^{\delta \tau}), \quad\text{$\P$-a.s.~for $y\ge x\ge 0$}.
\end{align}
\end{lemma}
\begin{proof}
From the integral form of \eqref{X1} (with $X^x_0=x$ and $X^y_0=y$), and noticing that $\pi(\cdot)$ is Lipschitz with constant $\delta>0$, it is immediate to see
\begin{align*}
\big| X^{x}_{\tau} - X^{y}_{\tau} \big|
& \le \big| x-y \big|+\int_{0}^{\tau}\big|\pi(X^{x}_{t})-\pi(X^{y}_{t}) \big| \ud t \le \big| x-y \big|+\delta \int_{0}^{\tau}\big|X^{x}_{t}-X^{y}_{t} \big| \ud t.
\end{align*}
Then, an application of Gronwall's inequality gives \eqref{eq:lipX}

The argument for \eqref{eq:bdX} is similar. Using that $y>x$ and $\pi(\cdot)$ Lipschitz we have
\begin{align*}
X^{y}_{\tau} - X^{x}_{\tau}& \ge y-x -\delta \int_{0}^{\tau}\big|X^{y}_{t}-X^{x}_{t} \big| \ud t  \notag\\
& \ge y-x -(y-x) \int_{0}^{\tau}\delta e^{\delta t}\ud t =(y-x)(2-e^{\delta \tau}),
\end{align*}
where the second inequality uses \eqref{eq:lipX}.
\end{proof}
The next estimate on the local time of the process $X$ is particularly useful to establish that the value function is Lipschitz in the time variable. In the rest of the paper we denote $L^z:=(L^z_t)_{t\in[0,T]}$ the local time of the process $X$ at a point $z\ge 0$, which is defined as (see, e.g., \cite[Eq.~(3.3.29), p.~68]{PS})
\begin{align}\label{def:lt}
L^z_t(X):=\lim_{\eps\downarrow 0}\frac{1}{2\eps}\int_0^t \mathds{1}_{\{|X_s-z|\le \eps\}}\ud \langle X\rangle_s,\quad\text{$\P$-a.s.}
\end{align}
Recall that $\E_x[\,\cdot\,]=\E[\,\cdot\,|X_0=x]$.
\begin{lemma}\label{lem:loctime}
Let $0< t_1\le t_2\le T$, fix $N> 0$ and recall $x_\alpha$ from \eqref{xalpha}. Then, there exists a positive constant $\kappa:=\kappa(t_1,N;x_\alpha)$ such that
\begin{align}\label{eq:L0}
\sup_{x\in[0,N]}\E_x\left[L^{x_\alpha}_{t_2}-L^{x_\alpha}_{t_1}\right]\le \kappa (t_2-t_1).
\end{align}
\end{lemma}
\begin{proof}
Thanks to \eqref{def:lt} we can select a sequence $(\eps_n)_{n\ge 1}$ such that $\eps_n\downarrow 0$ as $n\to \infty$ and
\begin{align*}
L^{x_\alpha}_{t_2}-L^{x_\alpha}_{t_1}=\lim_{n\to\infty}\frac{1}{2\eps_n}\int_{t_1}^{t_2} \mathds{1}_{\{|X_s-x_\alpha|\le \eps_n\}}\ud \langle X\rangle_s,\qquad\P_x-a.s.
\end{align*}
Then, using Fatou's lemma we get
\begin{align}\label{eq:L1}
\E_x\left[L^{x_\alpha}_{t_2}-L^{x_\alpha}_{t_1}\right]\le\liminf_{n\to\infty}\frac{1}{2\eps_n}\int_{t_1}^{t_2} \P_x\left(|X_s-x_\alpha|\le \eps_n\right)\sigma^2\ud s.
\end{align}
It is well-known that $X$ admits a transition density with respect to its speed measure (see, e.g., \cite[Thm.~50.11]{RW} or \cite[Sec.~4.11]{IMcK}). That is
\[
\P_x\left(|X_s-x_\alpha|\le \eps_n\right)=\int^{x_\alpha+\eps_n}_{x_\alpha-\eps_n}p(s,x,y)\frac{2\ud y}{\sigma^2 S'(y)},
\]
where $S'$ is the derivative of the scale function and reads
\begin{align}\label{eq:scale}
S'(y)=\exp\left(-\frac{2}{\sigma^2}\int_0^y\pi(z)\ud z\right).
\end{align}
Moreover, the map $(s,x,y)\mapsto p(s,x,y)$ is continuous on $(0,\infty)\times \R^2$ and clearly $S'$ is continuous too. Hence, letting $\eps_n\le \eps_0$, for all $n\ge 1$ and some $\eps_0>0$, and setting
\[
\kappa(t_1,N;x_\alpha):=2\sup_{(s,x,y)}\frac{p(s,x,y)}{S'(y)},
\]
with the supremum taken over $(s,x,y)\in[t_1,T]\times[0,N]\times[x_\alpha-\eps_0,x_\alpha+\eps_0]$, it is immediate to obtain \eqref{eq:L0} from \eqref{eq:L1}.
\end{proof}
\begin{remark}
Notice that in the lemma above $t_1$ must be taken strictly positive as the constant $\kappa(t_1,N;x_\alpha)$ might (and will) explode as $t_1\to 0$.
\end{remark}

\subsection{Continuity and monotonicity of the value function}
Some parts of the analysis in our paper are more conveniently performed by considering a different formulation of problem \eqref{valuev}.
Recall the infinitesimal generator $\cL$ of $X$ (see \eqref{eq:L}) and define the function
\begin{equation}\label{H}
H(x):=\begin{cases}
e^{-x}\left(\tfrac{\sigma^2 }{2} -\pi(x)\right), & x\le x_\alpha \\
\left(1-\gamma\right)e^{-x}\left(\tfrac{\sigma^2 }{2} -\pi(x)\right), & x>x_\alpha,
\end{cases}
\end{equation}
with $x_\alpha$ from \eqref{xalpha}. For future reference it is worth noticing that, since $r^G<r$,
\begin{align}\label{eq:bH}
-(r-r^G)\le H(x)\le \delta \quad\text{for $x\in\R_+$}.
\end{align}
Clearly $H$ is discontinuous at $x_\alpha$ and it is easy to check that $H(x)=(\cL h)(x)$ for $x\neq x_\alpha$ (recall that $h$ depends on $\alpha$).
Since $x\mapsto h(x)$ (see \eqref{defh}) is a convex function and its first derivative has a single jump
\[
\partial_x h(x_\alpha+)-\partial_xh(x_\alpha-)=\gamma \alpha,
\]
we can apply It\^o-Tanaka's formula to $h(X_{\tau\wedge\tau^{\dagger}})$ in \eqref{valuev}, to obtain the following equivalent formulation of problem \eqref{valuev}
\begin{align}\label{u}
u(t,x):=&\, v(t,x)-h(x)\\
=&\sup_{0 \leq \tau \leq T-t}  \E_{x} \bigg[ \int_{0}^{\tau\wedge\tau^\dagger} H(X_s) \mathds{1}_{\{X_s \neq x_\alpha\}} \ud s+\frac{\gamma\alpha}{2}
L^{x_\alpha}_{\tau\wedge\tau^\dagger}\bigg],\notag
\end{align}
where $(L^{z}_t)_{t\ge 0}$ is the local time of $X$ at a point $z>0$ (see \eqref{def:lt}).
Notice that $u$ is non-negative since $v(t,x) \geq h(x)$, for all $(t,x)\in[0,T]\times\R_+$, by \eqref{valuev}. We will show in Proposition \ref{prop1} that the presence of local time $L^{x_\alpha}$ in \eqref{u} implies that it is never optimal to stop when the process $X$ is equal to $x_\alpha$.

\begin{proposition}\label{prop-Vfun}
The following properties hold for the value function of the optimal stopping problem \eqref{valuev}:
\begin{itemize}
\item[  i)] the map $t\mapsto v(t,x)$ is decreasing and $v(T,x)=h(x)$ for any fixed $x\ge 0$;
\item[ ii)] the map $x\mapsto v(t,x)$ is decreasing and $v(t,0)=h(0)=1$ for any fixed $t\in[0,T]$.
\end{itemize}
Moreover, for any $0\le x_1\le x_2<+\infty$ and any $t\in[0,T]$ it holds
\begin{align}
\label{eq:cx}&0\le v(t,x_1)-v(t,x_2)\le \kappa_0 (x_2-x_1),
\end{align}
with $\kappa_0:=e^{\delta T}$. The map $t\mapsto v(t,x)$ is continuous on $[0,T]$ for any $x\in\R_+$ and, finally, for any $0\le t_1\le t_2< T$ and any $x\in[0,N]$, with fixed $N>0$, there is a constant $\kappa_1=\kappa_1(t_2,N;x_\alpha)>0$ such that
\begin{align}
\label{eq:ct}&0\le v(t_1,x)-v(t_2,x)\le \kappa_1(t_2-t_1).
\end{align}
\end{proposition}
\begin{proof}
The monotonicity in point $i)$ follows from time-independence of $h$ and $\tau^\dagger$, whereas the value of $v$ at $T$ follows from \eqref{valuev}. As for $ii)$, $v(t,0)=h(0)$ since $\tau^\dagger(0)=0$ $\P$-a.s. To show monotonicity of $v$ in $x$, fix $x_1<x_2$ and note that by uniqueness of the solution to \eqref{X1} follows $X^{x_1}_{s\wedge\tau^\dagger(x_1)}\le X^{x_2}_{s\wedge\tau^\dagger(x_2)}$ $\P$-a.s.~for all $s\in [0,T]$. Since the inequality also holds if we replace $s$ by a stopping time and the gain function $h$ is decreasing, we obtain
\begin{align}\label{tiz13}
v(t,x_2)-v(t,x_1)\le\sup_{0\le \tau\le T-t}\E\Big[h(X^{x_2}_{\tau\wedge\tau^\dagger(x_2)})-h(X^{x_1}_{\tau\wedge\tau^\dagger(x_1)})\Big]\le 0.
\end{align}

Next we prove \eqref{eq:cx}. Fix $t\in[0,T]$, consider $0\le x_1<x_2$ and denote by $\tau^{\dagger}_{1} := \tau^{\dagger}(x_{1})$ and  $\tau^{\dagger}_{2} := \tau^{\dagger}(x_{2})$ the first hitting time at zero of $X^{x_1}$ and $X^{x_2}$, respectively. From pathwise uniqueness of the solution of \eqref{X1} we have $\tau^{\dagger}_{1}\leq\tau^{\dagger}_{2}$. Then $\tau\wedge\tau^\dagger_1\wedge\tau^\dagger_2=\tau\wedge\tau^\dagger_1$, $\P$-a.s.~for every admissible stopping time $\tau$. Recalling that $v(t,\,\cdot\,)$ is decreasing and that $h$ is strictly decreasing and 1-Lipschitz (see \eqref{eq:ph}) we have
\begin{align}\hspace{-8pt}
0\leq &\, v(t,x_{1})-v(t,x_{2})\le \sup_{0\leq\tau\leq T-t} \E\Big[h\big(X^{x_{1}}_{\tau\wedge\tau^{\dagger}_{1}}\big)-h\big(X^{x_{2}}_{\tau\wedge
\tau^{\dagger}_{1}}\big)\Big]\\
\le&\,\E\Big[\sup_{0\leq s\leq T-t}\Big|X^{x_{1}}_{s}-X^{x_{2}}_{s}\Big|\Big]\le e^{\delta T}(x_2-x_1),\nonumber
\end{align}
where the first inequality is obtained by taking $\tau\wedge\tau^\dagger_1$ in $v(t,x_2)$ and the last inequality follows by \eqref{eq:lipX}.

It remains to prove \eqref{eq:ct}. For that it is convenient to use \eqref{u} and notice that for $0\le t_1\le t_2< T$ and $x\in\R_+$ we have
\[
0\le v(t_1,x)-v(t_2,x)=u(t_1,x)-u(t_2,x),
\]
where the inequality is due to $i)$ above. For any stopping time $\tau\in[0,T-t_1]$ we have that $\tau\wedge(T-t_{2})$ is admissible for the problem with value $u(t_2,x)$. Then, by direct comparison (recall that $\tau^\dagger$ only depends on $x\in\R_+$) and with $x\in[0,N]$, we have
\begin{align}\label{eq:lt1}
0\le& u(t_1,x)-u(t_2,x)\notag\\
\le& \sup_{0\le \tau\le T-t_1}\E_x\!\left[\mathds{1}_{\{\tau\wedge\tau^\dagger>T-t_2\}}\!\!\left(\!\int^{\tau\wedge\tau^\dagger}_{T-t_2}\!\!\!\mathds{1}_{\{X_s\neq x_\alpha\}}H(X_s)\ud s+\frac{\alpha\gamma}{2}\!\left(L^{x_\alpha}_{\tau\wedge\tau^\dagger}-L^{x_\alpha}_{T-t_2}\right)\!\right)\right]\\
\le &\, \delta (t_2-t_1)+\frac{\alpha\gamma}{2}\E_x\left[L^{x_\alpha}_{T-t_1}-L^{x_\alpha}_{T-t_2}\right],\notag
\end{align}
where in the final inequality we used \eqref{eq:bH} and the fact that the local time $t\mapsto L^{x_\alpha}_t$ is non-decreasing. Continuity of $t\mapsto u(t,x)$ (hence of $t\mapsto v(t,x)$) is now clear by continuity of local time's sample paths.  Recalling Lemma \ref{lem:loctime} we also obtain \eqref{eq:ct} by setting $\kappa_1(t_2,N;x_\alpha):=\delta+\alpha\gamma/2\cdot \kappa(T-t_2,N;x_\alpha)$.
\end{proof}
An immediate consequence of the proposition above, and the fact that $h$ is bounded and non-negative, is given in the next corollary.
\begin{corollary}\label{cont-vfun}
The value function $v$ of the optimal stopping problem \eqref{valuev} is non-negative, continuous on $[0,T]\times\R_+$ and bounded by 1.
\end{corollary}

Recalling the sets $\cC$ and $\cS$ defined in \eqref{setC0} and \eqref{setS0}, we can express them in terms of the function $u$ from \eqref{u} as
\begin{align*}
\mathcal{C}
= \big\{ (t,x)\in [0,T]\times \R_+ \ : \ u(t,x)>0 \big\}\notag
\end{align*}
and
\begin{align*}
\mathcal{S}
= \big\{ (t,x)\in [0,T]\times \R_+ \ : \ u(t,x)=0 \big\}.\notag
\end{align*}
Continuity of $v$ and $h$ imply that the sets $\cC$ and $\cS$ are open and closed, respectively. Moreover, \cite[Cor.\ 2.9, Sec.\ 2]{PS} guarantees that $\tau_*$ defined in \eqref{entry} is optimal for $v(t,x)$ for any $(t,x)\in[0,T]\times\R_+$. Finally, \cite[Thm.\ 2.4, Sec.\ 2]{PS} ensures that the process $V^{t,x}:=(V^{t,x}_s)_{s\in[0,T-t]}$ given by $V^{t,x}_s= v(t+s,X^x_s)$ is a supermartingale while $(V^{t,x}_{s\wedge \tau^*})_{s\in[0,T-t]}$ is a martingale, for any $(t,x)\in[0,T]\times\R_+$.
Using the martingale property and continuity of the value function we obtain the next well-known result (see, e.g.~\cite[Sec.~7.1, Chapter III]{PS}, for a proof).
\begin{proposition}\label{prop:PDE}
The value function $v$ lies in $C^{1,2}(\cC)$ and it solves the boundary value problem
\begin{align}\label{PDEu}
\partial_t v+\cL v = 0,\qquad \text{in $\cC$}
\end{align}
with $v=h$ on $\partial\cC$.
\end{proposition}

The next simple technical lemma is consequence of the maximum principle and it will be used later to prove continuity and strict monotonicity of the stopping boundary.
\begin{lemma}\label{lem:ut}
For all $(t,x)\in \cC$ it holds $\partial_t v(t,x)<0$.
\end{lemma}
\begin{proof}
By contradiction we assume there is $(t_0,x_0)\in\cC$ such that $\partial_t v(t_0,x_0)=0$. Since $v(t_0,x_0)>h(x_0)$ and $v(T,x_0)=h(x_0)$, by continuity there must exists $t_1\in(t_0,T)$ such that $(t_1,x_0)\in \cC$ and $\partial_t v(t_1,x_0)<-\eps$, for some $\eps>0$. By continuity of $\partial_t v$ inside $\cC$, and the fact that $\cC$ is open, there exists $\delta>0$ such that $\partial_t v(t_1,x)<-\eps/2$ for $x\in(x_0-\delta,x_0+\delta)$ and $\{t_1\}\times(x_0-\delta,x_0+\delta)\subset \cC$.

Now, letting $\cO:=(t_0,t_1)\times(x_0-\delta,x_0+\delta)$ we have that $\cO\subseteq\cC$ and $\partial_t v\in C^{1,2}(\cO)$, thanks to internal regularity results for solutions of partial differential equations applied to \eqref{PDEu} (see, e.g., \cite[Thm.~10, Ch.~3, Sec.~5]{friedman}). Moreover, differentiating \eqref{PDEu} with respect to time and using Proposition \ref{prop-Vfun}-(i) with the observations above, we obtain that $\hat v:= \partial_t v$ solves
\begin{align}
&(\partial_t\hat v+\cL \hat v)(t,x) = 0,\qquad\text{for $(t,x)\in\cO$}\\
&\hat v(t,x_0\pm\delta)\le 0,\qquad\qquad\text{for $t\in[t_0,t_1)$}\\
&\hat v(t_1,x)<-\eps/2,\qquad\qquad\text{for $x\in(x_0-\delta,x_0+\delta)$}.
\end{align}
Setting $\tau_\cO:=\inf\{s\ge 0\,:\,(t_0+s,X^{x_0}_s)\notin\cO\}$, an application of Dynkin's formula gives
\begin{align}
0=\hat v(t_0,x_0)=\E\left[\hat v (t_0+\tau_{\cO},X^{x_0}_{\tau_\cO})\right]\le -\frac{\eps}{2}\P\big(\tau_\cO=t_1-t_0\big),
\end{align}
which leads to a contradiction as the process $(t_0+s,X^{x_0}_s)$ exits $\cO$ by crossing the segment $\{t_1\}\times(x_0-\delta,x_0+\delta)$ with positive probability.
\end{proof}

It is clear by \eqref{u} that $u$ inherits the same continuity and boundedness properties of $v$ (see \eqref{eq:ct}, \eqref{eq:cx} and Corollary \ref{cont-vfun}). Moreover, $\partial_t u<0$ in $\cC$ with $u\in C^{1,2}$ in $\cC\setminus\left([0,T]\times\{x_\alpha\}\right)$ due to \eqref{defh} and \eqref{H}. Finally, in $\cC\setminus\left([0,T]\times\{x_\alpha\}\right)$ the function $u$ solves
\begin{align}\label{PDEu2}
\partial_t u+\cL u = -H,
\end{align}
with $u=0$ on $\partial\cC$.


\section{The free boundary problem}\label{sec:freeb}

In this section we study the free boundary problem associated with the stopping problem \eqref{u}. We derive geometric properties of the continuation region $\cC$ and regularity of its boundary $\partial\cC$. These have a close interplay with the smoothness of the value function $v$ in the whole space.

\subsection{Analysis of the stopping region} We start the study of the stopping region by noting that for any $(t,x)\in[0,T]\times\R_+$ it holds
\begin{equation}
\label{monot}
(t,x)\in \cS \Rightarrow [t,T]\times\{x\}\in \cS,
\end{equation}
since $t\mapsto u(t,x)$ is non-increasing (see (i) in Proposition \ref{prop-Vfun}). 

Some of the arguments that we need in order to characterise the stopping region require the next lemma. Its proof is somewhat standard but we provide it in the Appendix for completeness.
\begin{lemma}\label{lem:tech}
For $\eps>0$ define
\begin{align*}
\rho_{\varepsilon} := \inf\{s\ge 0\,:\, X^{x_\alpha}_{s} \notin (x_\alpha-\varepsilon ,x_\alpha+ \varepsilon)\}.
\end{align*}
Then, for any $\ell>0$ there exists $t_{\eps,\ell}>0$ such that
\[
\E_{x_\alpha} \left[L^{x_\alpha}_{s\wedge \rho_\varepsilon}\right]>\ell\, \E_{x_\alpha} \left[s\wedge \rho_\varepsilon \right]\quad \text{for all $s\in(0,t_{\eps,\ell})$.}
\]
\end{lemma}

Now we can use the lemma to show that it is never optimal to stop at $x_\alpha$.
\begin{proposition}\label{prop1}
It holds $[0,T)\times \{x_\alpha \} \subset \cC$.
\end{proposition}
\begin{proof}
Fix $\eps>0$ and let $\rho_\varepsilon$ be as in Lemma \ref{lem:tech}. Take $t\in[0,T)$ and $s\in [0,T-t)$. Since stopping at $s\wedge\rho_\eps$ is admissible for the problem with value function $u(t,x_\alpha)$, and
\[
\inf_{|\zeta|\le\eps} H(x_\alpha+\zeta)\ge -c_\varepsilon,
\]
for some $c_\eps>0$ only depending on $\eps$, one obtains
\begin{align*}
u(t,x_\alpha)& \ge \E_{x_\alpha} \bigg[ \int_{0}^{s\wedge \rho_\varepsilon} H(X_u) \mathds{1}_{\{X_u \neq x_\alpha\}} du+\frac{\gamma\alpha}{2}
L^{x_\alpha}_{s\wedge \rho_\varepsilon}\bigg]\\
& \ge  \tfrac{\gamma\alpha}{2} \E_{x_\alpha} \big[ L^{x_\alpha}_{s\wedge \rho_\varepsilon}\big]-c_{\varepsilon}\E_{x_\alpha} \big[s\wedge \rho_\varepsilon\big].
\end{align*}
Now, applying Lemma \ref{lem:tech} with $\ell=2c_\eps/(\gamma\alpha)$ and picking $s>0$ sufficiently small gives $u(t,x_\alpha)>0$. Hence $(t,x_\alpha)\in \cC$. Since $t\in[0,T)$ was arbitrary, the claim follows.
\end{proof}

For any initial point $(t,x)$ with $t\in[0,T)$ and $x\in\R_+\setminus\{x_\alpha\}$ such that $H(x)>0$, we can choose to stop at the first exit time from a small interval centered at $x$. Since $H>0$ in such interval, and the stopping time is strictly positive $\P$-a.s., this well-known argument gives $u(t,x)>0$. Then, it follows that $\mathcal{R}\subseteq\cC$, where
\begin{align}\label{R}
\mathcal{R}:=\{(t,x)\in[0,T)\times(\R_+ \backslash \{x_\alpha\}): H(x)>0\}.
\end{align}
Combining this observation with Proposition \eqref{prop1} we get
\begin{align}\label{eq:R}
\mathcal{R} \cup ([0,T)\times\{ x_\alpha\})\subseteq \cC.
\end{align}

It is clear that the shape of the set $\mathcal{R}$ varies depending on the parameters of the problem. Interestingly, this gives rise to two possible shapes of the stopping region, as we will see in the rest of the section. Let us start by noticing that
\begin{align}\label{x0bar}
H(x)>0 \ \Longleftrightarrow \ \frac{\sigma^2}{2}-\pi(x)>0 \ \Longleftrightarrow \ x>\bar x_0=\beta+\frac{r}{\delta},
\end{align}
where we used \eqref{pi} and $r>r_G$. Then, based on the fact that
\[
\cS\subseteq \big(\mathcal{R}^c\cap\{x\neq x_\alpha\}\big)\cup\big(\{T\}\times\R_+\big),
\]
where $\mathcal{R}^c$ is the complement of $\mathcal{R}$, we distinguish two cases:
\vspace{+5pt}

\noindent\textbf{Case 1}: $x_\alpha<\bar x_0$, then we have
\begin{align}\label{eq:case1}
\cS \subseteq \left\{[0,T)\times\Big([0,x_\alpha)\cup (x_\alpha,\bar x_0]\Big)\right\}\cup\Big(\{T\}\times\R_+\Big).
\end{align}
\vspace{+5pt}

\noindent\textbf{Case 2}: $x_\alpha\ge \bar x_0$, then we have
\begin{align}\label{eq:case2}
\cS \subseteq \Big\{[0,T)\times[0,\bar x_0)\Big\}\cup\Big(\{T\}\times\R_+\Big).
\end{align}
\vspace{+5pt}

We now focus on the study of the optimal stopping region in Case 1. Case 2 is easier and can be handled with simpler methods.
Thanks to \eqref{monot} we may write (recall \eqref{eq:c})
\begin{align}\label{def:c}
c(x)=\inf\{t\in[0,T]\,:\, u(t,x)=0\},\quad\text{for $x\in\R_+$.}
\end{align}
Notice that
\begin{align}\label{c-alpha}
c(x_\alpha)=T\qquad\text{and}\qquad c(x)=T,\quad\text{for $x>\bar x_0$},
\end{align}
due to \eqref{eq:R}. The next proposition is the main result in this subsection. It provides piecewise monotonicity and right/left-continuity of $c(\,\cdot\,)$.
\begin{proposition}\label{prop:monot-c2}
Assume $x_\alpha<\bar x_0$. The map $x\mapsto c(x)$ attains a global minimum $0\le \hat c\le T$ on $[x_\alpha,\bar x_0]$. Moreover, there exist $x_1\in[0,x_\alpha)$, $x_2\in(x_\alpha,\bar x_0)$ and $x_3\in[x_2,\bar x_0)$ such that $c(\cdot)$ is
\begin{itemize}
\item[(i)] equal to zero on $[0,x_1]$, strictly increasing on $(x_1,x_\alpha]$ and left-continuous on $[0,x_\alpha)$;
\item[(ii)] strictly decreasing on $[x_\alpha,x_2)$ and right-continuous on $[x_\alpha,x_3)$ ;
\item[(iii)] strictly increasing on $[x_3,\bar x_0)$ and left-continuous on $[x_2,\bar x_0)$.
\end{itemize}
(Notice that in $(ii)$ and $(iii)$ it might be $x_2=x_3$.)

In all cases $\hat c= c(x_2)=c(x_3)$ and, if $x_2<x_3$, then $\hat c=0$. Finally,
\begin{align}\label{eq:lim2}
\lim_{x\to x_\alpha} c(x)=c(x_\alpha)=T\quad\text{and}\quad
T=\lim_{x\downarrow \bar x_0} c(x)\ge c(\bar x_0)=\lim_{x\uparrow \bar x_0}c(x).
\end{align}
\end{proposition}

The proof relies on two technical lemmas which we are going to present first.

\begin{lemma}\label{prop:strip}
Assume $x_\alpha<\bar x_0$. Then
\begin{itemize}
\item[(i)] for $z<x_\alpha$ and $t\in(0,T]$ it holds
\[
(t,z)\in\cS\implies [t,T]\times[0,z]\subseteq \cS;
\]
\item[(ii)] for $z_1,z_2\in(x_\alpha,\bar x_0)$, with $z_1<z_2$, and $t\in(0,T]$ it holds
\[
(t,z_1),(t,z_2)\in\cS\implies[t,T]\times[z_1,z_2]\subseteq\cS.
\]
\end{itemize}
The map $x\mapsto c(x)$ is never strictly positive and constant (simultaneously) on intervals $(z_1,z_2)$ contained in $[0,x_\alpha)\cup(x_\alpha,\bar x_0)$. Finally, for every interval $(z_1,z_2)$ contained in $[0,x_\alpha)\cup(x_\alpha,\bar x_0)$ it holds
\begin{align}\label{nonempty}
\cS\cap\big((0,T)\times(z_1,z_2)\big)\neq \varnothing.
\end{align}
\end{lemma}
\begin{proof}
First we prove (i) and (ii). The two claims are similar since $(t,0)\in\cS$ for all $t\in[0,T]$ (see Proposition \ref{prop-Vfun}-(ii)).  Then it is enough to show (ii) as the proof of point (i) is analogous up to obvious changes.

Let $(t,z_1)$ and $(t,z_2)$ belong to $\cS$ and $x_\alpha<z_1<z_2\le \bar x_0$. If $t=T$ the result is trivial due to \eqref{c-alpha}. Then let $t<T$ and recall that $H(x)<0$ for $x\in(x_\alpha,\bar x_0)$. By \eqref{monot} we know that $[t,T]\times\{z_i\}\subseteq\cS$ for $i=1,2$. Then it suffices to show that also $\{t\}\times(z_1,z_2)\subseteq\cS$. Arguing by contradiction assume there exists $z_3\in(z_1,z_2)$ such that $(t,z_3)\in\cC$. Let $\tau^*_3=\tau^*(t,z_3)$ be optimal for the problem with value $u(t,z_3)$. Then
\begin{align*}
u(t,z_3)=  \E_{z_3} \bigg[ \int_{0}^{\tau^*_3\wedge\tau^\dagger} H(X_s) \mathds{1}_{\{X_s \neq x_\alpha\}} \ud s+\frac{\gamma\alpha}{2} L^{x_\alpha}_{\tau^*_3\wedge\tau^\dagger}\bigg].
\end{align*}
Since $[t,T]\times \{z_i\} \subseteq \cS$ for $i=1,2$, we have that
\[
\tau^*_3\le \zeta:=\inf\{s\ge 0\,:\,(t+s,X^{z_3}_s)\notin[t,T)\times(z_1,z_2)\}.
\]
Hence, $u(t,z_3)< 0$ because for all $s\ge 0$ we have $L^{x_\alpha}_{s\wedge\zeta}=0$ and $H(X_{s\wedge\zeta})<0$, $\P_{z_3}$-a.s.~and $\P_{z_3}(\tau^*_3>0)=1$ by assumption. Thus we have a contradiction.

Next we show that $c$ cannot be strictly positive and constant. The proof borrows some ideas from \cite{DeA15}. Arguing by contradiction, assume that there exists an interval $(z_1,z_2)\subset[0,x_\alpha)\cup (x_\alpha,\bar x_0)$ where $c(x)$ takes the constant value $\bar c>0$. Then the open set $\cO:=(0,\bar c)\times(z_1,z_2)$ is contained in $\cC$ and $u\in C^{1,2}(\cO)$ since so are $v$, by Proposition \ref{prop:PDE}, and $h$ (away from $x_\alpha$). It follows that $u$ satisfies
\begin{align}\label{PDE0}
\left\{
  \begin{array}{ll}
    \partial_t u+\mathcal{L}u=-H, & \hbox{\textrm{in} $\cO$,} \\
    u(\bar c,x)=0, & \hbox{$x \in (z_1,z_2)$ .}
  \end{array}
\right.
\end{align}
Pick $\varphi \in C^{\infty}_{c}(z_1,z_2)$, with $\varphi \ge 0$. Thanks to \eqref{PDE0}, for $s \in [0,\bar c)$ we have
\begin{align*}
\int_{z_1}^{z_2}\partial_t u(s,y) \varphi(y) \ud y &= -\int_{z_1}^{z_2}(\mathcal{L}u)(s,y) \varphi(y) \ud y - \int_{z_1}^{z_2}H(y) \varphi(y) \ud y\\
&= \int_{z_1}^{z_2}u(s,y)(\mathcal{L^*}\varphi)(y) \ud y - \int_{z_1}^{z_2}H(y) \varphi(y) \ud y\,
\end{align*}
where we used integration by parts and $\mathcal{L^*}$ is the adjoint operator of $\mathcal{L}$. Recalling that $u_t \leq 0$ (Proposition \ref{prop-Vfun}-(i)), we use dominated convergence to obtain
\begin{align}\label{contrPDE}
0 &\ge  \lim_{s\uparrow \bar c} \int_{z_1}^{z_2}u(s,y)(\mathcal{L^*}\varphi)(y) \ud y- \int_{z_1}^{z_2}H(y) \varphi(y) \ud y\notag\\
&=\int_{z_1}^{z_2}\lim_{s\uparrow \bar c} u(s,y)(\mathcal{L^*}\varphi)(y) \ud y- \int_{z_1}^{z_2}H(y) \varphi(y) \ud y\\
&= - \int_{z_1}^{z_2}H(y) \varphi(y) \ud y >0,\notag
\end{align}
where the last equality is due to $u(\bar c,y)=0$ and the final inequality follows from the facts that $H<0$ on $(0,\bar x_0)$ and $\varphi$ is arbitrary. Hence a contradiction.

Finally, by the same argument we can prove \eqref{nonempty}. Indeed, if $\cO\!:=\!(0,T)\!\times\!(z_1,z_2)\!\subseteq\! \cC$ for some interval $(z_1,z_2)\!\subset\! [0,x_\alpha)\cup(x_\alpha,\bar x_0)$. That would imply $c(x)\!=\!T$ on $(z_1,z_2)$, contradicting that $c$ cannot be strictly positive and constant.
\end{proof}

\begin{lemma}\label{lem:tech2}
Assume $x_\alpha<\bar x_0$. The map $x\mapsto c(x)$ is lower semi-continuous on $\R_+$ and it is continuous at $x_\alpha$ with $c(x_\alpha)=T$. Moreover 
\begin{align}\label{eq:c<T}
c(x)<T\quad\text{for $x\in[0,x_\alpha)\cup(x_\alpha,\bar x_0)$.}
\end{align}
\end{lemma}
\begin{proof}
Recall that $c(x)=T$ for $x\in (\bar x_0,\infty)$ by \eqref{c-alpha}. Thus $c(\cdot)$ is continuous on $(\bar x_0,\infty)$.  

Now fix $z\in (0,\bar x_0]$ and take a sequence $(z_n)_{n\ge 1}\subseteq (0,\infty)$ with $z_n\to z$ as $n\to \infty$. Then
\[
\liminf_{n\to\infty}\big(c(z_n),z_n\big)=\big(\liminf_{n\to \infty}c(z_n),z\big),
\]
and since $(c(z_n),z_n)_{n\ge 1}\subseteq \cS$ and $\cS$ is closed, it must be
\[
\big(\liminf_{n\to \infty}c(z_n),z\big)\in\cS.
\]
The latter implies $\liminf_{n\to \infty}c(z_n)\ge c(z)$, by definition of $c(z)$, and lower semi-continuity follows.

If $z=0$, then $c(0)=0$ by Proposition \ref{prop-Vfun}-(ii). Obviously $\liminf_{n\to\infty}c(z_n)\ge 0$ for any $z_n\to 0$, hence lower semi-continuity holds at $z=0$ too. To prove continuity at $x_\alpha$ recall that $c(x_\alpha)=T$ by Proposition \ref{prop1}, then $\liminf_{n\to\infty} c(z_n)\ge c(x_\alpha)=T$ together with $c(\cdot)\le T$, imply
\[
\liminf_{n\to\infty}c(z_n)=T=\limsup_{n\to\infty}c(z_n)
\]
for any $z_n\to x_\alpha$. 

It remains to prove \eqref{eq:c<T}. Since $c(0)=0$ it suffices to assume there is $x\in (0,\bar x_0)\setminus\{x_\alpha\}$ such that $c(x)=T$ and then argue by contradiction. With no loss of generality assume $x\in(x_\alpha,\bar x_0)$ since the case of $x\in(0,x_\alpha)$ can be treated analogously. Then we can pick $z_1,z_2\in(x_\alpha,\bar x_0)$ such that $z_1<x<z_2$ and $\bar t:=c(z_1)\vee c(z_2)<T$, by \eqref{nonempty}. Hence (ii) in Lemma \ref{prop:strip} implies $(\bar t, x)\in\cS$, i.e., $c(x)\le \bar t$, which is a contradiction.
\end{proof}

\begin{proof}[{\em\bf Proof of Proposition \ref{prop:monot-c2}}]
From (i) in Lemma \ref{prop:strip} we immediately deduce that $x\mapsto c(x)$ is non-decreasing on $[0,x_\alpha)$. Moreover, \eqref{eq:c<T} and the fact that $c(\,\cdot\,)$ cannot be strictly positive and constant (Lemma \ref{prop:strip}) also guarantee that there exists $x_1\in[0,x_\alpha)$ such that $c(x)=0$ on $[0,x_1]$ and $c(\,\cdot\,)$ is strictly increasing on $(x_1,x_\alpha]$ (notice that it could be $x_1=0$ and $c(\,\cdot\,)>0$ on $(0,x_\alpha)$). Left-continuity of $c$ on $[0,x_\alpha)$ follows by its monotonicity and lower semi-continuity.

By lower semi-continuity on $\R_+$ and \eqref{eq:c<T} there must be a minimum of $c(\,\cdot\,)$ on $[x_\alpha,\bar x_0]$, denoted $\hat c\in[0,T)$. We have two possible cases: either $\hat c=0$ or $\hat c>0$.
\begin{itemize}
\item[(a)] If $\hat c=0$, then the minimum may occur at most on an interval $[x_2,x_3]\subseteq (x_\alpha,\bar x_0]$. Indeed, the $\mathrm{argmin}_{[x_\alpha,\bar x_0]} c(x)$ is closed by lower semi-continuity of $c(\,\cdot\,)$ and it must be connected by (ii) in Lemma \ref{prop:strip}. However  the interval $[x_2,x_3]$ may collapse into a single point $x_2=x_3$ (in which case $c(\cdot)>0$ on $[x_\alpha,x_2)\cup(x_2,\bar x_0]$);
\item[(b)] If $\hat c>0$, then it may only occur at a single point $x_2(=x_3)\in(x_\alpha,\bar x_0]$, again by (ii) in Lemma \ref{prop:strip} and since $c(\,\cdot\,)$ cannot be strictly positive and constant.
\end{itemize}

Strict monotonicity on $[x_\alpha,x_2)$ and $(x_3,\bar x_0]$ now follows from (ii) in Lemma \ref{prop:strip} and the fact that $c(\,\cdot\,)$ cannot be strictly positive and constant. Left/right-continuity are then obtained by monotonicity and lower semi-continuity.
Finally, the first limit in \eqref{eq:lim2} follows by Lemma \ref{lem:tech2}, whereas the second one is trivial.
\end{proof}

By arguments as above we obtain analogous results for the case of $x_\alpha\ge \bar x_0$. Therefore we omit the proof of the next proposition.
\begin{proposition}\label{prop:monot3}
Assume $x_\alpha\ge \bar x_0$. Then, on the interval $[0,\bar x_0)$ the map $x\mapsto c(x)$ is non-decreasing, left-continuous, with $c(x)<T$. On the interval $[\bar x_0,+\infty)$ it holds $c(x)=T$ and
\begin{align}\label{c-lim1}
\lim_{x\uparrow \bar x_0}c(x)=c(\bar x_0)\le T.
\end{align}
Moreover, there exists at most a point $x_1\le \bar x_0$ such that $c(x)=0$ for $x\in[0,x_1]$ and $c(\cdot)$ is strictly increasing on $(x_1,\bar x_0]$.
\end{proposition}

\subsection{Higher regularity of the value function and of the optimal boundary}
Thanks to the geometry of the optimal boundary we obtain a lemma that will be used to establish global $C^1$-regularity of the value function (jointly in $(t,x)$).

As shown in \cite{DeAPe18} the key to $C^1$-regularity of the value function is the probabilistic regularity of the stopping boundary. Since the 2-dimensional process $(t,X_t)_{t\ge 0}$ is not of strong Feller type, we will actually use probabilistic regularity for the {\em interior} $\cS^\circ$ of the stopping region. For completeness we recall that a process $Z\in\R^d$ is said to be of strong Feller type if $z\mapsto \E_z[f(Z_t)]$ is continuous for any $t>0$ and any {\em bounded measurable} function $f:\R^d\to \R$.

More precisely, letting
\begin{align*}
&\sigma_*(t,x):=\inf\{s\in(0,T-t]:(t+s,X^x_s)\in\cS\},\quad\text{and}\\
&\sigma^\circ_*(t,x):=\inf\{s\in(0,T-t]:(t+s,X^x_s)\in\cS^\circ\}.
\end{align*}
we say that a boundary point $(t,x)\in\partial\cC$ is (probabilistically) regular for $\cS$ (or $\cS^\circ$) if
\begin{align}\label{def:reg-b}
\P(\sigma_*(t,x)=0)=1\qquad(\text{or}~\P(\sigma^\circ_*(t,x)=0)=1).
\end{align}
Clearly, probabilistic regularity for $\cS^\circ$ implies the one for $\cS$. However, regularity for $\cS^\circ$ is meaningless at points $(t_0,z_0)\in\partial\cC$ such that $\cS$ has empty interior in a neighbourhood of $(t_0,z_0)$. Therefore in what follows we need both.

\begin{lemma}\label{lem:reg}
The boundary $\partial\cC$ is probabilistically regular for $\cS$. Moreover, for any $(t_0,z_0)\in \partial\cC$ and any sequence $(t_n,x_n) \to(t_0,z_0)$ as $n\to\infty$, it holds
\begin{align}\label{eq:taulim}
\lim_{n \to \infty}\tau^*(t_n,x_n) =0, \quad \P\textrm{-a.s.}
\end{align}
where $\tau^*$ is defined in \eqref{entry}.
\end{lemma}
\begin{proof}
By the law of iterated logarithm and the geometry of the stopping region, it is clear that
\begin{align}\label{eq:sst}
\sigma_*(t,z)=\sigma^\circ_*(t,z)=\tau^*(t,z),\quad \text{$\P$-a.s.}
\end{align}
for all $(t,z)\in\partial\cC$ except at most along vertical stretches of the boundary corresponding to $x_1=0$ and $x_2=x_3$, as defined in Proposition \ref{prop:monot-c2}. Indeed, at such points a {\em spike} may occur so that $\cS^\circ$ may be (locally) empty.
For simplicity let us denote
\[
\cE:=\Big(\big(0,c(x_1+)\big)\times\{x_1\}\Big)\cup \Big(\big(\hat c,c(x_2-)\wedge c(x_2+)\big)\times\{x_2\}\Big).
\]
By definition $\tau^*(t,z)=0$, $\P$-a.s., for all $(t,z)\in\partial \cC$. Then, by \eqref{eq:sst} we have regularity of $\partial\cC\setminus\cE$ for $\cS^\circ$ in the sense of \eqref{def:reg-b}. Hence \eqref{eq:taulim} holds for any $(t,z)\in\partial\cC\setminus\cE$ (see, e.g., Corollary 6 in \cite{DeAPe18}).

Thanks to lower semi-continuity of $c$, it only remains to consider regularity at $\cE$ in the cases: (a) $x_2=x_3$ but $\hat c<c(x_2\pm)$, and (b) $x_1=0$ but $c(x_1+)>0$. We give a full argument for case (a), then case (b) may be handled analogously.

Let us assume $x_2=x_3$ but $\hat c<c(x_2\pm)$. Then $\sigma_*(t,x_2)=\tau^*(t, x_2)$, $\P$-a.s., continues to hold for all $t\in[0,T)$ such that $(t,x_2)\in\partial\cC$, by the law of iterated logarithm. Hence the first in \eqref{def:reg-b} holds. Since the hitting time $\sigma^\circ_*(t,x_2)$ is no longer zero for $\hat c \le t <c(x_2+)\wedge c(x_2-)$, because there is no interior part to the stopping region in a neighbourhood of $(t,x_2)$, the argument provided in \cite{DeAPe18} to prove the analogue of \eqref{eq:taulim} needs a small tweak.

Fix $(t_0,x_2)\in\partial \cC$ with $\hat c \le t_0 <c(x_2+)\wedge c(x_2-)$ and a sequence $(t_n,x_n)_{n\ge 1}\subset \cC$ that converges to $(t_0,x_2)$ as $n\to\infty$. Recall that we work with a continuous modification of the stochastic flow and let us pick $\omega\in\Omega$ outside a null set such that $(t,x)\mapsto X^x_t(\omega)$ is continuous. Then for any $\delta>0$, there exist $0<s_{1,\omega}<s_{2,\omega}<\delta$ such that $X^{x_2}_{s_{1,\omega}}(\omega)<x_2<X^{x_2}_{s_{2,\omega}}(\omega)$, by the law of iterated logarithm. By continuity of $x\mapsto X^x$, for some $N_{\delta,\omega}\ge 1$ and  all $n\ge N_{\delta,\omega}$, we have $X^{x_n}_{s_{1,\omega}}(\omega)<x_2<X^{x_n}_{s_{2,\omega}}(\omega)$. With no loss of generality we may assume that $N_{\delta,\omega}$ is sufficiently large that $t_n\ge t_0-s_{1,\omega}$ for $n\ge N_{\delta,\omega}$. Hence, the points $(t_n+s_{1,\omega},X^{x_n}_{s_{1,\omega}}(\omega))$ and $(t_n+s_{2,\omega},X^{x_n}_{s_{2,\omega}}(\omega))$ lie in the two opposite half-planes that are adjacent to the segment $[t_0,T]\times\{x_2\}$. This implies that for each $n\ge N_{\delta,\omega}$ there is $s_{n,\omega}\in(s_{1,\omega},s_{2,\omega})$ such that $X^{x_n}_{s_{n,\omega}}(\omega)=x_2$ and therefore $(t_n+s_{n,\omega}, X^{x_n}_{s_{n,\omega}}(\omega))\in\cS$. The latter implies $\tau^*(t_n,x_n)(\omega)\le \delta$ for all $n\ge N_{\delta,\omega}$, hence
\[
\limsup_{n\to\infty}\tau^*(t_n,x_n)(\omega)\le \delta.
\]
Since $\delta>0$ and $\omega$ were arbitrary we obtain \eqref{eq:taulim}.
\end{proof}

We now provide some useful estimates for $\partial_x v$ in $\cC$. Below we use that $(t_0,z_0)\in\partial\cC$ with $t_0<T$ guarantees $z_0 \neq x_{\alpha}$ by Proposition \ref{prop1}, hence $\partial_x h$ is continuous at $z_0$. In particular,
$\partial_x h\left( X^{x}_{\tau^*}\right)$ is well-defined on the event $\{\tau^*<T-t\}$, for any $(t,x)\in [0,T)\times\R_+$.

\begin{lemma}\label{lem:vx}
For all $(t,x)\in \cC$ and $0<s<(\delta^{-1}\ln(2))\wedge (T-t)$ it holds
\begin{align}\label{eqvx}
e^{\delta s}\big(&\E \left[\mathds{1}_{\{\tau^*\le s\}}\partial_x h\left( X^x_{\tau^*}\right)\right]-\kappa_0 \P\left( \tau^*>s \right)\big)\\
&\le \partial_xv(t,x)\le\big(2-e^{\delta s}\big)\E \left[\mathds{1}_{\{\tau^*< s\wedge\tau^\dagger\}}\partial_x h\left( X^x_{\tau^*}\right) \right],\notag
\end{align}
with $\tau^*=\tau^*(t,x)$ and $\tau^\dagger=\tau^\dagger(x)$.
\end{lemma}
\begin{proof}
Recall that for any initial condition $(t,x)\in[0,T]\times\R_+$ the process $V^{t,x}_s= v(t+s,X^x_s)$ is a continuous supermartingale and $s\mapsto V^{t,x}_{s\wedge\tau^*}$ is a continuous martingale for $s\in[0,T-t]$.
Fix $(t,x)\in \cC$ and $\varepsilon>0$ such that $(t,x+\varepsilon)\in \cC$ and $(t,x-\varepsilon)\in \cC$. Notice that
\[
\tau^{\dagger}(x-\varepsilon)\le \tau^{\dagger}(x)\le \tau^{\dagger}(x+\varepsilon),\:\:\:\text{$\P$-a.s.}
\]
and, by (ii) in Proposition \ref{prop-Vfun}, that $\tau^\dagger(x)\ge \tau^*(t,x)$ a.s., because $[0,T]\times \{0 \} \subseteq \cS$. Set $\tau^*:=\tau^*(t,x)$ to simplify notation. Then for all $s<T-t$, using the (super)martingale property, we have
\begin{align*}
v(t,x+\varepsilon) &\ge \E \left[v\big(t+(s\wedge\tau^*),X^{x+\varepsilon}_{s\wedge\tau^*}\big) \right]\\
v(t,x) &= \E \left[v\big(t+(s\wedge\tau^*),X^{x}_{s\wedge\tau^*}\big) \right].
\end{align*}
Thus
\begin{align}\label{eq:vx1}
&v(t,x+\varepsilon)- v(t,x)\notag\\
&\ge \E \left[v\big(t+(s\wedge\tau^*),X^{x+\varepsilon}_{s\wedge\tau^*}\big)-v\big(t+(s\wedge\tau^*),X^{x}_{s\wedge\tau^*}\big) \right]\notag\\
 &= \E \left[\mathds{1}_{\{\tau^*\le s\}}\Big(v\big(t+\tau^*,X^{x+\varepsilon}_{\tau^*}\big)-v\big(t+\tau^*,X^{x}_{\tau^*}\big)\Big) \right]\\
&\hspace{+25pt}+ \E \left[\mathds{1}_{\{\tau^*> s\}}\Big(v\left(t+s,X^{x+\varepsilon}_{s}\right)-v\left(t+s,X^{x}_{s}\right)\Big) \right]\notag\\
&\ge \E \left[\mathds{1}_{\{\tau^*\le s\}}\Big(h\left(X^{x+\varepsilon}_{\tau^*}\right)-h\left(X^{x}_{\tau^*}\right)\Big) \right]-\kappa_0 \E \Big[\mathds{1}_{\{\tau^*> s\}}\big|X^{x+\varepsilon}_{s}-X^{x}_{s}\big| \Big],\notag
\end{align}
where $\kappa_0>0$ is as in \eqref{eq:cx}.

On $\{\tau^*\le s\}$, the decreasing property of $h$ and $\big|X^{x+\varepsilon}_{s}-X^{x}_{s}\big| \le \varepsilon e^{\delta s}$ (see \eqref{eq:lipX}) give $h\left(X^{x+\varepsilon}_{\tau^*}\right)\ge h\left(X^{x}_{\tau^*}+\varepsilon e^{\delta s}\right)$. Hence
\begin{align*}
v(t,x+\varepsilon)- v(t,x)&\ge  \E \left[\mathds{1}_{\{\tau^*\le s\}}\left(h\big(X^{x}_{\tau^*}+\varepsilon e^{\delta s}\big)-h\big(X^{x}_{\tau^*}\big)\right) \right]-\kappa_0\,\varepsilon\, e^{\delta s} \P\left(\tau^*> s \right)\\
&= \E \left[\mathds{1}_{\{\tau^*\le s\}}\int_{0}^{\varepsilon e^{\delta s}}\partial_xh\left( X^{x}_{\tau^*}+z\right)
\ud z \right]-\kappa_0\,\varepsilon\, e^{\delta s} \P\left(\tau^*> s \right),
\end{align*}
since $h$ is absolutely continuous on $\R_+$. Then
\begin{align*}
\partial_x v(t,x)&= \lim_{\varepsilon \downarrow 0} \frac{1}{\varepsilon}\Big(v(t,x+\varepsilon)- v(t,x)\Big)\\
& \ge \lim_{\varepsilon\downarrow 0}  \E \left[\mathds{1}_{\{\tau^*\le s\}}\frac{1}{\varepsilon} \int_{0}^{\varepsilon e^{\delta s}}\partial_x h\left( X^{x}_{\tau^*}+z\right)
\ud z \right]-\kappa_0 e^{\delta s}\P\left( \tau^*>s \right)\\
& = \E \left[\mathds{1}_{\{\tau^*\le s\}}\lim_{\varepsilon \downarrow 0} \frac{1}{\varepsilon} \int_{0}^{\varepsilon e^{\delta s}}\partial_x h\left( X^{x}_{\tau^*}+z\right)
\ud z \right]-\kappa_0 e^{\delta s}\P\left( \tau^*>s \right),
\end{align*}
where the final equality follows by dominated convergence since $\big|\partial_x h\big|\le 1$.

Now, for each $\omega \in \{\tau^*\le s\}$ we have $X^{x}_{\tau^*}(\omega)\neq x_\alpha$ by Proposition \ref{prop1} since $s<T-t$. Hence, there exists $\bar{\varepsilon}_{\omega}>0$ such that the mapping
$z \mapsto \partial_x h\left( X^{x}_{\tau^*}(\omega)+z\right)$ is continuous on $\left[0,\bar{\varepsilon}_{\omega}e^{\delta s}\right]$ and an application of the fundamental theorem of calculus gives
\begin{align}\label{eq:FC}
\lim_{\varepsilon \downarrow 0} \frac{1}{\varepsilon} \int_{0}^{\varepsilon e^{\delta s}}\partial_x h\left( X^{x}_{\tau^*}(\omega)+z\right)
\ud z=e^{\delta s}\partial_x h\left( X^{x}_{\tau^*}(\omega)\right).
\end{align}
Hence
\[
\partial_xv(t,x)\ge  e^{\delta s}\left(\E \left[\mathds{1}_{\{\tau^*\le s\}}\partial_x h\left( X^{x}_{\tau^*}\right)\right]-\kappa_0\P\left( \tau^*>s \right)\right).
\]

Next we want to bound from above the difference $v(t,x)-v(t,x-\varepsilon)$. This requires a slight modification of the previous argument in order to account for the fact that $\tau^\dagger(x-\eps)\le \tau^\dagger(x)$, a.s. In particular, with no loss of generality we assume that $\eps\in(0,\eps_0]$ for some $\eps_0>0$ fixed.
Letting $\tau^\dagger_0:=\tau^\dagger(x-\eps_0)$ for simplicity, we have $\tau^\dagger_0\le \tau^\dagger(x-\eps)$. Then, arguing as in \eqref{eq:vx1} gives
\begin{align*}
&v(t,x)-v(t,x-\varepsilon)\\
&\le \E \left[v\left(t+(s\wedge\tau^*\wedge\tau^\dagger_0),X^{x}_{s\wedge\tau^*\wedge\tau^\dagger_0}\right)-v\left(t+(s\wedge\tau^*\wedge\tau^\dagger_0),X^{x-\varepsilon}_{s\wedge\tau^*\wedge\tau^\dagger_0}\right) \right]\\
&\le \E \left[\mathds{1}_{\{\tau^*\le s\wedge\tau^\dagger_0\}}\left(h\left(X^{x}_{\tau^*}\right)-h\left(X^{x-\varepsilon}_{\tau^*}\right)\right) \right]\\
&\hspace{+25pt}+ \E \left[\mathds{1}_{\{\tau^*> s\wedge\tau^\dagger_0\}}\left(v\left(t+(s\wedge\tau^\dagger_0),X^{x}_{s\wedge\tau^\dagger_0}\right)-v\left(t+(s\wedge\tau^\dagger_0),X^{x-\varepsilon}_{s\wedge\tau^\dagger_0}\right)\right) \right].
\end{align*}
Notice that the second term in the last expression is negative thanks to (ii)-Proposition \ref{prop-Vfun}.
Moreover, on the event $\{\tau^*\le s\wedge\tau^\dagger_0\}$ we have
\[
h\left(X^{x}_{\tau^*}\right)-h\left(X^{x-\varepsilon}_{\tau^*}\right)=\int_{0}^{X^{x}_{\tau^*}-X^{x-\varepsilon}_{\tau^*}}\partial_x h\left( X^{x-\varepsilon}_{\tau^*}+z\right)\ud z\le \int_{0}^{X^{x}_{\tau^*}-X^{x-\varepsilon}_{\tau^*}}\partial_x h\left( X^{x}_{\tau^*}+z\right) \ud z
\]
where the last step follows from the convexity of $h(\cdot)$. Also, on the event $\{\tau^*\le s\wedge\tau^\dagger_0\}$, using \eqref{eq:bdX} we have
\[
X^{x}_{\tau^*}-X^{x-\varepsilon}_{\tau^*}\ge \varepsilon \left(2-e^{\delta \tau^*}\right)\ge \varepsilon \left(2-e^{\delta s}\right)> 0,
\]
by assuming $s < \delta^{-1}\ln(2)$ with no loss of generality.
It follows that, since $\partial_x h\le 0$, we have
\[
h\left(X^{x}_{\tau^*}\right)-h\left(X^{x-\varepsilon}_{\tau^*}\right)\le \int_{0}^{\varepsilon \left(2-e^{\delta s}\right)}\partial_x h\left( X^{x}_{\tau^*}+z\right) \ud z
\]
on $\{\tau^*\le s\wedge\tau^\dagger_0\}$.
Thus
\[
v(t,x)-v(t,x-\varepsilon)\le \E \left[\mathds{1}_{\{\tau^*\le s\wedge\tau^\dagger_0\}}\int_{0}^{\varepsilon \left(2-e^{\delta s}\right)}\partial_x h\left( X^{x}_{\tau^*}+z\right) \ud z \right],
\]
and, by arguments as in \eqref{eq:FC}, we obtain
\[
\partial_x v(t,x)\le \left(2-e^{\delta s}\right)\E \left[\mathds{1}_{\{\tau^*\le s\wedge\tau^\dagger_0\}}\partial_x h\left( X^{x}_{\tau^*}\right) \right], \ \ \ \ \textrm{for }  s<(\delta^{-1}\ln(2))\wedge (T-t).
\]
To conclude we let $\eps_0\downarrow 0$ so that $\tau^\dagger_0=\tau^\dagger(x-\eps_0)\uparrow\tau^\dagger(x)$ and the upper bound in \eqref{eqvx} holds by monotone convergence.
\end{proof}

\begin{proposition}\label{prop:gradient}
Fix any $(t_0,z_0)\in \partial \cC $ with $t_0<T$ and $z_0>0$. Then, for any sequence $(t_n,x_n)_{n\ge 1}\subset\cC$ such that $(t_n,x_n)\rightarrow (t_0,z_0)$ as $n\uparrow\infty$,
we have
\begin{align}\label{vx}
\lim_{n \rightarrow\infty}\partial_x v(t_n,x_n)=\partial_x h(z_0)
\end{align}
and
\begin{align}\label{vt}
\lim_{n \rightarrow\infty}\partial_t v(t_n,x_n)=0.
\end{align}
\end{proposition}
\begin{proof}
First we prove \eqref{vx}. Notice, that \eqref{eqvx} holds for any point $(t_n,x_n)\!\in\! \cC$ from a sequence that converges to $(t_0,z_0)$ as $n\uparrow\infty$, where $(t_0,z_0)\in\partial\cC$ with $t_0<T$ and $z_0>0$.
Since $\tau^{*}_n:=\tau^{*}(t_n,x_n) \rightarrow 0$ as $n\uparrow\infty$, by Lemma \ref{lem:reg}, and $\partial_x h\left( X^{x_n}_{\tau^{*}_n}\right)\rightarrow \partial_x h(z_0)$ (recall that $z_0\neq x_{\alpha}$), then for $0<s<(\delta^{-1}\ln(2))\wedge (T-t_0)$ dominated convergence and \eqref{eqvx} give
\[
e^{\delta s} \partial_x h(z_0)\le \liminf_{n\to \infty} \partial_x v(t_n,x_n)
\le \limsup_{n\to \infty} \partial_x v(t_n,x_n) \le (2-e^{\delta s})\partial_x h(z_0).
\]
Letting $s \rightarrow 0$ we get \eqref{vx}.

To prove \eqref{vt}, fix $(t,x)\in \cC$ with $t<T$ and $\varepsilon>0$ such that $(t+\varepsilon,x)\in \cC$. Let
\[
\tau_N:=\inf\{u\ge 0: X^x_u\ge N\},
\]
and pick $s<T-(t+\varepsilon)$.
Then, proceeding as in the proof of Lemma \ref{lem:vx}, with $\tau^*=\tau^*(t,x)$ we have
\begin{align*}
v(t+\varepsilon,x) &\ge \E \left[v\big(t+\varepsilon+(s\wedge\tau^{*}\wedge \tau_N),X^{x}_{s\wedge\tau^{*}\wedge \tau_N}\big) \right],\\
v(t,x) &= \E \left[v\big(t+(s\wedge\tau^{*}\wedge \tau_N),X^{x}_{s\wedge\tau^{*}\wedge \tau_N}\big) \right].
\end{align*}
Combining the above with \eqref{eq:ct} gives
\begin{align*}
&v(t+\varepsilon,x)- v(t,x)\\
&\ge \E \left[\mathds{1}_{\{\tau^{*}\le s\wedge\tau_N\}}\Big(h\left(X^{x}_{\tau^{*}}\right)-h\left(X^{x}_{\tau^{*}}\right)\Big) \right]\\
&\hspace{+25pt}+ \E \left[\mathds{1}_{\{\tau^{*}> s\wedge\tau_N\}}\Big(v\big(t+\varepsilon+(s\wedge\tau_N),X^{x}_{s\wedge\tau_N}\big)-v\big(t+(s\wedge\tau_N),X^{x}_{s\wedge\tau_N}\big)\Big) \right]\\
&\ge -\kappa_1(t+\eps+s,N) \,\varepsilon \,\P \left(\tau^{*}> s\wedge\tau_N\right).
\end{align*}
With no loss of generality we may assume that $\eps\in(0,\eps_0]$ for some $\eps_0>0$ such that $s<T-(t+\eps_0)$ and $\kappa_1(t+\eps+s,N)\le \hat \kappa_1(\eps_0,N)$ for some constant $\hat \kappa_1(\eps_0,N)>0$. Hence,
\[
0\ge \partial_t v(t,x) \ge -\hat \kappa_1(\eps_0,N)\, \P \left(\tau^{*}> s\wedge\tau_N \right).
\]
The result holds for any $(t_n,x_n)\in \cC$ from a sequence converging to $(t_0,x_0)$. Moreover, with no loss of generality we can assume that $x_n\le x_0+1<N$ for all $n\ge 1$ so that $\tau_N(x_n)\ge \tau_N(x_0+1)>0$, $\P$-a.s., and
\[
\P \Big(\tau^{*}(t_n,x_n)> s\wedge\tau_N(x_n) \Big)\le \P \Big(\tau^{*}(t_n,x_n)> s\wedge\tau_N(x_0+1) \Big).
\]
Then, thanks to Lemma \ref{lem:reg} we get
\[
0 \ge  \limsup_{n\to \infty} \partial_t v (t_n,x_n) \ge \liminf_{n\to \infty} \partial_t v (t_n,x_n)\ge 0.
\]
Hence \eqref{vt} holds.
\end{proof}

Proposition  \ref{prop:gradient} and Proposition \ref{prop:PDE} imply continuous differentiability of $v$.
\begin{corollary}\label{cor:C1reg}
The value function $v$ is continuously differentiable on the set $[0,T)\times (0,+\infty)$. Moreover, $v\in C^{1,2}\big(\,\overline\cC\cap\big([0,T)\times(0,+\infty)\big)\big)$ with
\begin{align}\label{eq:vxx}
\lim_{\cC\ni(t,x)\to(t_0,z_0)\in\partial\cC}\partial_{xx}v(t,x)=-\frac{2}{\sigma^2}\pi(z_0)\partial_x h(z_0)
\end{align}
for $t_0<T$ and $z_0>0$.
\end{corollary}
\begin{proof}
We only need to prove \eqref{eq:vxx}. In order to do that it is sufficient to take limits in \eqref{PDEu} and use Proposition \ref{prop:gradient}.
\end{proof}
The next theorem shows that the optimal boundary is continuous as a function of $x$. Notice that this type of continuity is not a standard result in optimal stopping problems for time-space processes $(t,X)$. Indeed, in the probabilistic literature, one normally proves continuity of the boundary as a function of {\em time}. Our proof relies on the use of a suitably constructed reflecting diffusion.
\begin{theorem}\label{prop:cont}
The mapping $x\mapsto c(x)$ is continuous on $(0,\infty)$. If $c(0+)=0$ then continuity holds on $\R_+$ (recall $\R_+=[0,\infty)$).
\end{theorem}
\begin{proof}
We give a full proof in the case $x_\alpha<\bar x_0$ and consider the interval $[x_3,\infty)$, with $x_3$ as in Proposition \ref{prop:monot-c2}, where the boundary is increasing. It will be clear that the intervals $[x_1,x_\alpha]$ and $[x_\alpha, x_2]$ and the case $x_\alpha\ge \bar x_0$ can be treated analogously.

Arguing by contradiction let us assume that there exists $z_0\in[x_3,\infty)$ such that $c(z_0)<c(z_0+)$ and let $\mathcal{I}_0:=(c(z_0),c(z_0+))$. Then $\cI_0\times\{z_0\}\subset\partial\cC$ and there exists $z_1>z_0$ such that $\partial_t u(t,z_1)<-\eps_1$ for some $\eps_1>0$ and for all $t\in\cI_1:=(t_0,t_1)\subset\cI_0$ for some $t_1>t_0$ (see Lemma \ref{lem:ut}).

Since $u\in C^1([0,T) \times \R_+)$ by Corollary \ref{cor:C1reg} and $\cI_1\times\{z_0\}\subset\partial\cC$, we have $\partial_t u(t,z_0)=\partial_x u(t,z_0)=0$ for $t \in \mathcal{I}_1$. Then for any $\eps>0$ there exists $\delta_{\varepsilon}>0$ such that
$z_0+\delta_\eps<z_1$ and
\begin{align}\label{eq:bdd}
0 \ge\partial_t u \ge -\varepsilon\quad\text{and}\quad |\partial_x u| \leq \varepsilon\quad
\text{on $\overline{\mathcal{I}_1} \times [z_0,z_0+\delta_{\varepsilon}]$,}
\end{align}
by uniform continuity on any compact.

Now we consider a process that equals $(X_t)_{t \ge 0}$ away from $z_0+\delta_{\varepsilon}$ and is reflected (upwards) at $z_0+\delta_{\varepsilon}$. It is well-known (see, e.g., \cite{LS84} or \cite[Sec.~12, Chapter I]{Bass}) that there exists a unique strong solution of the stochastic differential equation
\begin{align*}
\ud \XX_t^{\varepsilon}=\pi(\XX_t^{\varepsilon}) \ud t + \sigma  \ud W_{t}+\ud A^{\delta_{\varepsilon}}_t, \qquad\qquad
\XX_{0}^{\varepsilon} = z_0+\delta_{\varepsilon},
\end{align*}
where $A^{\delta_{\varepsilon}}$ is a continuous, non-decreasing process that guarantees, $\P$-a.s.,
\begin{align}\label{eq:SK}
&\XX_t^{\varepsilon} \ge z_0+\delta_{\varepsilon},\quad\text{for all $t\ge 0$ and}\qquad\int_{0}^T\mathds{1}_{\{\XX^\eps_t>z_0+\delta\}}\ud A^{\delta_\eps}_t=0.
\end{align}

As in Lemma \ref{lem:ut} we appeal to classical results on interior regularity for solutions of PDEs that guarantee $\partial_t u\in C^{1,2}\big(\cI_1\times(z_0,z_1)\big)$ and $(\partial_t+\cL)\partial_t u=0$ on $\cI_1\times(z_0,z_1)$. Then, setting $\tau^\eps_1:=\inf\{s\ge0\,:\,\XX^{\varepsilon}_s =z_1\}$ and $\hat u:= \partial_t u$, an application of It\^o's formula for semi-martingales gives, for any $t\in\mathcal{I}_1$
\begin{align}\label{eq:cont0}
\E&\! \left[ \hat u(t+\tau^\eps_1\wedge(t_1-t),\XX_{\tau^\eps_1\wedge(t_1-t)}^{\varepsilon})\right]\notag\\
    =& \hat u(t,z_0+\delta_{\eps})+\E\left[\int_0^{\tau^\eps_1\wedge(t_1-t)}\partial_x \hat u(t+s,\XX_{s}^{\varepsilon}) \,\ud A^{\delta_\eps}_{s}\right]\\
\ge& - \eps +\E\left[\int_0^{\tau^\eps_1\wedge(t_1-t)}\partial_{tx}u(t+s,z_0+\delta_{\varepsilon})\, \ud A_{s}^{\delta_{\varepsilon}}\right]\notag
\end{align}
where the inequality follows from \eqref{eq:bdd} and the second condition in \eqref{eq:SK} implies
\[
\ud A_{s}^{\delta_{\varepsilon}}=\mathds{1}_{\{\XX^\eps_s=z_0+\delta_\eps\}}\ud A_{s}^{\delta_{\varepsilon}}.
\]
For the expression on the left-hand side of \eqref{eq:cont0} we have
\begin{align*}
\E& \left[ \hat u(t+\tau^\eps_1\wedge(t_1-t),\XX_{\tau^\eps_1\wedge(t_1-t)}^{\varepsilon})\right]\\
\leq& \E\left[\mathds{1}_{\{\tau^\eps_1<t_1-t\}}\hat u(t+\tau^\eps_1,z_1)\right] \leq -\eps_1\, \P (\tau^\eps_1<t_1-t).
\end{align*}
Hence, from \eqref{eq:cont0} we obtain
\begin{align}\label{eq:cont1}
-\eps_1 \P (\tau^\eps_1<t_1-t) \ge - \eps +\E \left[\int_0^{\tau^\eps_1\wedge(t_1-t)}\partial_{tx}u(t+s,z_0+\delta_{\varepsilon})\, \ud A_{s}^{\delta_{\varepsilon}}\right].
\end{align}

The next step is to let $\eps\to 0$. However, the regularity of $\partial_{tx}u$ as $\delta_\eps\downarrow 0$ might be problematic. We therefore use a trick with test functions to overcome this difficulty. Pick $\varphi \in C^{\infty}_{c}(\mathcal{I}_1)$, $\varphi \ge 0$ such that
$\int_{\mathcal{I}_1}\varphi(t) \ dt = 1$. Then, multiplying both sides of \eqref{eq:cont1} by $\varphi$, integrating over $\cI_1$ and using Fubini's theorem we obtain
\begin{align*}
-\eps_1 & \int_{\mathcal{I}_1} \P(\tau^\eps_1<t_1-t)\varphi(t) \ud t\\
& \ge - \eps +\E\left[\int_0^{\tau^\eps_1}\left(\int_{\mathcal{I}_1}\mathds{1}_{\{t<t_1-s\}}\partial_{tx}u(t+s,z_0+\delta_{\varepsilon})\varphi(t) \ud t \right) \ud A_{s}^{\delta_{\varepsilon}}\right]\\
& = - \eps +\E\bigg[\int_0^{\tau^\eps_1}\bigg(\partial_x u(t_1,z_0+\delta_{\varepsilon})\varphi(t_1-s)\\
&\qquad\qquad\qquad\qquad-\int_{\mathcal{I}_1}\mathds{1}_{\{t<t_1-s\}}\partial_xu(t+s,z_0+\delta_{\varepsilon})\varphi'(t) \ud t \bigg) \ud A_{s}^{\delta_{\varepsilon}}\bigg]\\
& \ge - \eps -\eps\, \E\bigg[\int_0^{\tau^\eps_1}\varphi(t_1-s) \ud A_{s}^{\delta_{\varepsilon}}-A^{\delta_\eps}_{\tau^\eps_1\wedge t_1}\int_{\cI_1}|\varphi'(t)|\ud t\bigg] \\
&\ge -\eps\left(1+\left(\|\varphi \|_{\infty}+T\|\varphi' \|_{\infty}\right)\E\big[A_{\tau^\eps_1\wedge t_1}^{\delta_{\varepsilon}}\big]\right),
\end{align*}
where for the penultimate inequality we have used the bounds on $\partial_x u$ given in \eqref{eq:bdd}, and the final inequality uses that $\varphi(t_1-s)=0$ for $s\ge t_1$. Here $\|\cdot\|_\infty$ is the supremum norm on $[0,T]$.

For the increasing process $A^{\delta_\eps}$ we have an upper bound which is independent of $\eps$. This can be deduced from the integral form of the SDE. That is, taking expectation of
\begin{align*}
\XX_{\tau^\eps_1\wedge t_1}^{\eps}= z_0+\delta_{\eps}+\int_0^{\tau^\eps_1\wedge t_1}\pi(\XX_{s}^{\eps})\ud s + \sigma W_{\tau^\eps_1\wedge t_1}+A^{\delta_{\varepsilon}}_{\tau^\eps_1\wedge t_1}
\end{align*}
gives
\begin{align*}
\E\left[A^{\delta_{\varepsilon}}_{\tau^\eps_1\wedge t_1}\right]=\E\left[\XX_{\tau^\eps_1\wedge t_1}^{\eps}
 -z_0-\delta_{\eps}-\int_0^{\tau^\eps_1\wedge t_1}\pi(\XX_{s}^{\eps})\ud s\right].
\end{align*}
Using that $\XX_{s}^{\eps}\in[z_0+\delta_{\eps},z_1]$ and $\pi(\XX_{s}^{\eps})\ge \pi(z_1)$ for $s \leq \tau^\eps_1$, we obtain
\begin{align*}
\E\left[A^{\delta_{\varepsilon}}_{\tau^\eps_1\wedge t_1}\right]\leq c_0:=z_1-z_0-t_1\, \pi(z_1).
\end{align*}
Hence
\begin{align*}
-\eps_1 & \int_{\mathcal{I}_1} \P(\tau^\eps_1<t_1-t)\varphi(t) \ud t \ge -\eps\left(1+\left(\|\varphi \|_{\infty}+\|\varphi' \|_{\infty}\right)c_0\right)
\end{align*}
and taking limits as $\eps\to 0$ gives
\begin{align}\label{eq:posit}
\limsup_{\eps\to 0} \int_{\mathcal{I}_1} \P(\tau^\eps_1<t_1-t)\varphi(t) \ud t \le 0.
\end{align}
If we can show that the left hand side above is positive we have reached a contradiction, and there cannot be a discontinuity of $c$.

For our final task we introduce the change of measure
\begin{align}
\frac{\ud\PP^\eps}{\ud \P}\bigg|_{\cF_T}:=\exp\left(-\int_0^T\sigma^{-1}\pi(\XX^\eps_s)\ud W_s-\frac{1}{2}\int_0^T\sigma^{-2}\pi^2(\XX^\eps_s)\ud s\right),
\end{align}
so that under $\PP^\eps$, we have
\[
\XX^\eps_t=z_0+\delta_\eps +\sigma\widehat W^\eps_t+A^{\delta_\eps}_t,
\]
where $\widehat W^\eps:=(\widehat W^\eps_t)_{t\in[0,T]}$ is a Brownian motion defined as
\[
\widehat W^\eps_t=W_t+\int_0^t\sigma^{-1}\pi(\XX^\eps_s)\ud s.
\]
For future reference we also introduce the D\'oleans-Dade exponential
\begin{align}\label{eq:Z}
Z^\eps_T:=\exp\left(\int_0^T\sigma^{-1}\pi(\XX^\eps_s)\ud \widehat W^\eps_s-\frac{1}{2}\int_0^T\sigma^{-2}\pi^2(\XX^\eps_s)\ud s\right).
\end{align}
Since the measures are equivalent on $\cF_T$, under $\PP^\eps$ the process $\XX^\eps$ is a Brownian motion reflected at $z_0+\delta_\eps$. Hence, we have an explicit formula for the increasing process $A^{\delta_\eps}$ (see, \cite[Lemma 6.14, Chapter 3]{KS1}), that is
\begin{align}\label{eq:A}
A^{\delta_\eps}_t=\sup_{0\le s\le t}\big(-\sigma \widehat W^\eps_s\big).
\end{align}

It remains to remove the dependence of the measure on $\eps$. For that we can take a filtered probability space $(\widehat \Omega, \widehat \cF,(\widehat \cF_t)_{t\in[0,T]}, \PP)$ equipped with a standard Brownian motion $B:=(B_t)_{t\in[0,T]}$. On such space we construct a Brownian motion starting from $z_0+\delta_\eps$ and reflected at its starting point, that we denote $Y$. That is
\[
Y^{\eps}_t=z_0+\delta_\eps +\sigma B_t+L_t,\qquad t\in[0,T],
\]
where $L$ takes the same expression of \eqref{eq:A} but with $B$ instead of $\widehat W^\eps$. For future reference we also denote
\[
Y^{0}_t=z_0+\sigma B_t+L_t,\qquad t\in[0,T].
\]
By construction
\[
\mathsf{Law}(\XX^\eps\,|\,\PP^\eps)=\mathsf{Law}(Y^{\eps}\,|\,\PP).
\]
Then, setting $\rho^\eps_1:=\inf\{t\ge 0:Y^\eps_t=z_1\}$, denoting $\EE^\eps$ the expectation under $\PP^\eps$ and letting $\xi^\eps_T$ be defined as the D\'oleans-Dade exponential in \eqref{eq:Z} but with $(Y^\eps,B)$ instead of $(\XX^\eps,\widehat W^\eps)$, we obtain
\begin{align}\label{prob}
\P(\tau^\eps_1<t_1-t)=\EE^\eps\left[Z^\eps_T\mathds{1}_{\{\tau^\eps_1<t_1-t\}}\right]=\EE\left[\xi^\eps_T\mathds{1}_{\{\rho^\eps_1<t_1-t\}}\right].
\end{align}
Using the explicit form of $Y^\eps$, under $\PP$ we have
\[
\lim_{\eps\to 0}\rho^\eps_1=\rho^0_1:=\inf\{s\ge 0:z_0+\sigma B_t+L_t=z_1 \},
\]
where the convergence is monotonic from above and therefore also
\[
\lim_{\eps\to 0}\mathds{1}_{\{\rho^\eps_1<t_1-t\}}=\mathds{1}_{\{\rho^0_1<t_1-t\}}.
\]
Hence, Fatou's lemma and \eqref{prob} give
\begin{align*}
\liminf_{\eps \downarrow 0} \P(\tau^\eps_1<t_1-t)\ge&  \EE\left[\liminf_{\eps\to 0} \xi^\eps_T \mathds{1}_{\{\rho^\eps_1<t_1-t\}}\right]=\widehat \E\left[\xi^0_T\mathds{1}_{\{\rho^0_1<t_1-t\}}\right]>0,
\end{align*}
where $\xi^0_T$ is the D\'oleans-Dade exponential associated to $(Y^0,B)$, and the final inequality follows from well-known distributional properties of reflected Brownian motion (see, e.g., \cite[Sec.~2.8.B]{KS1}).

Finally, using Fatou's lemma in \eqref{eq:posit}, and the discussion above, we conclude
\begin{align*}
0\ge \liminf_{\eps \downarrow 0} \int_{\mathcal{I}_1}\P(\tau^\eps_1<t_1-t)\varphi(t) \ud t \ge \int_{\mathcal{I}_1} \widehat \E\left[\xi^0_T\mathds{1}_{\{\rho^0_1<t_1-t\}}\right]\varphi(t) \ud t>0,
\end{align*}
where the final inequality uses that $\varphi\ge 0$ and arbitrary. Hence a contradiction and continuity of $x\mapsto c(x)$ is proved.
\end{proof}
If $c(0+)>0$ the smooth-fit may break down on $\{0\}\times\big[0,c(0+)\big)$, i.e., $\partial_x u(t,0+)\neq 0$ for some $t\in\big[0,c(0+)\big)$. That is why continuity of $c(\,\cdot\,)$ only holds on $(0,\infty)$ in that case. Combining the continuity result with \eqref{eq:lim2} and \eqref{c-lim1} also guarantees:
\begin{corollary}\label{cor:lim-c}
It holds
\[
\lim_{x\to\bar x_0}c(x)=T.
\]
\end{corollary}

We conclude the section by giving the proofs of the main results stated in Section \ref{sec:mainres}.

\subsection{Proofs of Theorems \ref{thm:v} and \ref{thm:boundary}}\label{sec:proof}

\begin{proof}[Proof of Theorem \ref{thm:v}]
The first claim is consequence of Corollary \ref{cont-vfun}, and $v\ge h$ follows by taking $\tau=0$ in \eqref{valuev}. Monotonicity of the mappings $t\mapsto v(t,x)$ and $x\mapsto v(t,x)$ was proven in Proposition \ref{prop-Vfun}. Continuous differentiability of $v$ and continuity of $\partial_{xx}v$ on $\overline\cC\cap\big([0,T)\times(0,\infty)\big)$ were obtained in Corollary \ref{cor:C1reg}. Proposition \ref{prop:PDE} guarantees that $v$ solves \eqref{PDEu0} in $\cC$. Moreover, $v=h$ in $\cS$ and $\cL h(x) \le 0$ in $\cS$ by \eqref{eq:R}. Hence \eqref{PDEu0} holds. 

As for uniqueness, if we can find another function $w$ that solves \eqref{PDEu0} with $\cC=\{w>h\}$ and with the same regularity as $v$, then by a standard verification argument based on a well-known generalisation of It\^o's formula we obtain that $w$ coincides with the value function of the optimal stopping problem \eqref{valuev}. Further details in this direction are omitted as they are standard and can be found in \cite[Thm.\ 4.2, Ch.\ IV]{BL}.
\end{proof}

\begin{proof}[Proof of Theorem \ref{thm:boundary}] We only provide the full argument for (b), which builds on the results of Proposition \ref{prop:monot-c2}. The proof of (a) is easier and follows from analogous arguments and Proposition \ref{prop:monot3}.

If $x_\alpha< \bar x_0$ our boundary $c$ is strictly monotonic and continuous on the intervals $[x_1,x_\alpha)$, $(x_\alpha,x_2)$ and $(x_3,\bar x_0)$  (with $x_1$, $x_2$ and $x_3$ as in Proposition \ref{prop:monot-c2}) and $c(x)=T$ for $x\ge \bar x_0$.
Hence, the map $x\mapsto c(x)$ can be inverted separately on the intervals $[x_1,x_\alpha)$, $(x_\alpha,x_2)$ and $(x_3,\bar x_0]$, to obtain three continuous and strictly monotonic functions of time that describe the boundary $\partial\cC$.

Recalling $\hat c$ from Proposition \ref{prop:monot-c2}, we can define locally the inverse functions
\begin{align}
\label{inv1} &b_1(t):=\inf\{x\in[0,x_\alpha):c(x)>t\},\qquad\:\:\:\: t\in[0,T),\\
\label{inv2} &b_2(t):=\sup\{x\in(x_\alpha,x_2):c(x)>t\},\qquad t\in[\hat c,T),\\
\label{inv3} &b_3(t):=\inf\{x\in(x_3,\bar x_0):c(x)>t\},\qquad\:\: t\in[\hat c,T).
\end{align}
Clearly $0\le b_1(t)\le b_2(t)\le b_3(t)$ for $t\in[\hat c,T)$. Setting $t_0:=c(0+)$ we have $b_1(t)=0$ on $[0,t_0)$ if $t_0>0$, otherwise $b(0+)=x_1$. Recalling that $x\mapsto c(x)$ is strictly increasing and continuous on $[x_1,x_\alpha]$ and $[x_3,\bar x_0]$, we obtain that $t\mapsto b_1(t)$ and $t\mapsto b_3(t)$ are strictly increasing and continuous on $[t_0,T)$ and $[\hat c, T)$ respectively. Analogously, since $x\mapsto c(x)$ is strictly decreasing and continuous on $[x_\alpha,x_2]$ we have $t\mapsto b_2(t)$ strictly decreasing and continuous on $[\hat c,T)$. Moreover, $b_1(T-)=b_2(T-)=x_\alpha$ by the first limit in \eqref{eq:lim2} and $b_3(T-)=\bar x_0$ by Corollary \ref{cor:lim-c}.  By construction, if $\hat c>0$ we have $b_2(\hat c)= b_3(\hat c)$ because $\hat c$ is the unique point in $\mathrm{argmin}_{[x_\alpha,\bar x_0]} c(x)$. If instead $\hat c=0$, we have $b_2(0)=x_2$ and $b_3(0)=x_3$ by Proposition \ref{prop:monot-c2}. This concludes the proof of i)--iv).

The claims in v) are now straightforward.
\end{proof}


\section{Some comments about management fees}\label{sec:fees}
In a model with proportional management fees as in \eqref{v:fees} the only change in our analysis is due to the fact that the function $H$ appearing in $u$ (see \eqref{H} and \eqref{u}) is replaced by
\begin{equation*}
H^{p,q}(x):=\begin{cases}
e^{-x}\left(\tfrac{\sigma^2 }{2}-q -\pi(x)\right)-p, & x\le x_\alpha, \\
\left(1-\gamma\right)e^{-x}\left(\tfrac{\sigma^2 }{2}-\tfrac{q}{1-\gamma} -\pi(x)\right)-p, & x>x_\alpha,
\end{cases}
\end{equation*}
where $p,q\ge 0$ are as in \eqref{v:fees}.

The key features in our analysis from the previous sections are the presence of the local time at $x_\alpha$, in the problem formulation \eqref{u}, and the sign of the function $H$ via \eqref{eq:R} (whose interplay leads to the two separate cases $x_\alpha<\bar x_0$ and $x_\alpha\ge \bar x_0$).
Here, the local time produces the same effects, so that Proposition \ref{prop1} continues to hold and $[0,T)\times\{x_\alpha\}\subset\cC$. Instead, we now need to look at the sign of the function $H^{p,q}$ rather than that of $H$.

It is immediate to check that, for $x\le x_\alpha,$
\[
H^{p,q}(x)>0\implies \frac{\sigma^2}{2}-q-\pi(x)>0\iff x>\bar x_{q}:=\beta+\frac{r+q}{\delta}
\]
and, if $x_\alpha<\bar x_q$, then $H^{p,q}(x)<0$. Similarly, for $x>x_\alpha$, 
\[
H^{p,q}(x)>0\implies \frac{\sigma^2}{2}-\frac{q}{1-\gamma}-\pi(x)>0\iff x>\bar x_{q,\gamma}:=\beta+\frac{r+q(1-\gamma)^{-1}}{\delta},
\]
where $\bar x_{q,\gamma}>\bar x_q>x_G$ since $\gamma\in(0,1)$ (recall $x_G=\beta+r^G/\delta$). Hence, if $x_\alpha<\bar x_q$, it is $H^{p,q}(x)<0$ for $x_\alpha<x<\bar x_{q,\gamma}$.

So, also in the presence of management fees we must consider various cases depending on the position of $x_\alpha$ relative to $\bar x_q$ and $\bar x_{q,\gamma}$. In keeping with the rest of the paper and in the interest of length, here we briefly illustrate only the case $x_\alpha<\bar x_q$ and draw a parallel with the case $x_\alpha<\bar x_0$ from the previous sections. The remaining cases can be studied analogously with the methods developed above.

From now on, let us assume $x_\alpha<\bar x_q$. Then
\begin{align}\label{eq:xqp}
H^{p,q}(x)>0\iff x-\bar x_{q,\gamma}>\frac{p}{\delta}(1-\gamma)^{-1}e^x.
\end{align}
In the special case $p=0$, we have the exact analogue of (b) in Theorem \ref{thm:boundary} but with $\bar x_0$ therein, replaced by $\bar x_{q,\gamma}$. If instead $p>0$ only two sub-cases may arise and we must consider them separately:
\vspace{+5pt}

\noindent{\bf Case (i)}: $x-\bar x_{q,\gamma}\le \frac{p}{\delta}(1-\gamma)^{-1}e^x$ for all $x\in\R_+$ (that is $\ln[\tfrac{\delta}{p}(1-\gamma)]\le 1+\bar x_{q,\gamma}$). Then
\[
\cS \subseteq \left\{[0,T)\times\Big([0,x_\alpha)\cup (x_\alpha,\infty)\Big)\right\}\cup\Big(\{T\}\times\R_+\Big);
\]
\vspace{+5pt}

\noindent{\bf Case (ii)}: There exist $\hat x_2>\hat x_1>\bar x_{q,\gamma}$ such that $x-\bar x_{q,\gamma}> \frac{p}{\delta}(1-\gamma)^{-1}e^x$ for all $x\in(\hat x_1,\hat x_2)$ and $x-\bar x_{q,\gamma}\le \frac{p}{\delta}(1-\gamma)^{-1}e^x$ otherwise (that is $\ln[\tfrac{\delta}{p}(1-\gamma)]> 1+\bar x_{q,\gamma}$). Then
\[
\cS \subseteq \left\{[0,T)\times\Big([0,x_\alpha)\cup (x_\alpha,\hat x_1]\cup[\hat x_2,\infty)\Big)\right\}\cup\Big(\{T\}\times\R_+\Big).
\]
\vspace{+5pt}

In case (ii) the situation is similar to \eqref{eq:case1}, with $\bar x_0$ therein replaced by $\hat x_1$ and noting the additional strip $[0,T]\times[\hat x_2,\infty)$ intersecting the stopping set. In particular, repeating the same arguments as in Proposition \ref{prop:monot-c2} one can prove that the map $x\mapsto c(x)$ (defined as in \eqref{def:c}) satisfies all claims in the proposition with $\bar x_0$ replaced by $\hat x_1$ throughout. In addition to that, and by the same methods, one can also prove that $c(x)<T$ for all $x>\hat x_2$ and there exists $\hat x_3>\hat x_2$ such that $x\mapsto c(x)$ is strictly decreasing on $(\hat x_2,\hat x_3)$ with $c(\hat x_2+)=T$ and $c(x)=0$ for $x\in[\hat x_3,\infty)$. Thanks to piece-wise monotonicity of the boundary, also in this setting we can prove its continuity as in Theorem \ref{prop:cont}.

While the strict monotonicity follows by the exact same arguments as those used in the proof of Proposition \ref{prop:monot-c2}, for completeness we prove the existence of $\hat x_3$, which did not appear in the previous analysis.

\begin{proposition} \label{prop:MF}
In the setting of {\em Case (ii)}, there exists $\hat x_3>\hat x_2$ such that $c(x)=0$ for $x\in[\hat x_3,\infty)$.
\end{proposition}
\begin{proof}
Let us argue by contradiction and assume $c(x)>0$ for all $x\ge \hat x_2$. In particular, let us first assume the stronger requirement that there exists $\theta>0$ such that $c(x)\ge \theta$ for all $x\ge \hat x_2$. Consider the value function $v(0,x)$ for $x>m>\hat x_2$ and a fixed $m$. Setting $\tau_{m}=\inf\{t\ge 0:X^{x}_t\le m\}$ and letting $\tau^*$ be optimal for $v(0,x)$, we have $\tau^*\ge \theta\wedge \tau_{m}$, $\P$-a.s. Now, using this observation we have
\begin{align*}
v(0,x)=&\E_x\Big[h(X_{\tau^*\wedge\tau^\dagger})-\int_0^{\tau^*\wedge\tau^\dagger}\Big(p+q e^{-X_t}\Big)\ud t\Big]\\
=&\E_x\Big[\mathds{1}_{\{\tau^*< \tau_m\}}\Big(h(X_{\tau^*\wedge\tau^\dagger})-\int_0^{\tau^*\wedge\tau^\dagger}\Big(p+q e^{-X_t}\Big)\ud t\Big)\Big]\\
&+\E_x\Big[\mathds{1}_{\{\tau^*\ge \tau_m\}}\Big(h(X_{\tau^*\wedge\tau^\dagger})-\int_0^{\tau^*\wedge\tau^\dagger}\Big(p+q e^{-X_t}\Big)\ud t\Big)\Big].
\end{align*}
On the event $\{\tau^*<\tau_m\}$ we have $\tau^*<\tau^\dagger$ and $\tau^*\ge \theta$. Moreover, on that event $h(X^x_{\tau^*\wedge\tau^\dagger})\le h(m)$ since $h$ is decreasing (see \eqref{eq:ph}) and $X^x_{\tau^*\wedge\tau^\dagger}\ge m$. Recalling that $h$ is positive and bounded by $1$, we obtain the upper bound
\begin{align*}
v(0,x)\le &\E_x\Big[\mathds{1}_{\{\tau^*< \tau_m\}}\Big(h(m)-p\theta\Big)\Big]+\P_x\big(\tau^*\ge \tau_m\big)\\
=& h(m)-p\theta-\E_x\Big[\mathds{1}_{\{\tau^*\ge \tau_m\}}\Big(h(m)-p\theta\Big)\Big]+\P_x\big(\tau^*\ge \tau_m\big)\\
\le & h(m)-p\theta+(p\theta+1)\P_x\big(\tau_m\le T\big).
\end{align*}
Letting $x\uparrow \infty$ we have $\P_x\big(\tau_m\le T\big)\downarrow 0$. Then, letting $m\uparrow\infty$ we also have $h(m)\downarrow \alpha\gamma$, so that
\[
\lim_{x\to\infty}v(0,x)\le \alpha\gamma-p\theta.
\]
The latter contradicts $v(0,x)\ge h(x)\ge \alpha\gamma$ for all $x\in\R_+$ (recall \eqref{defh}). Hence, it cannot be $c(x)\ge \theta$ on $[\hat x_2,\infty)$.

Now we prove that indeed it cannot be $c(x)>0$ on $[\hat x_2,\infty)$. By way of contradiction, assume the latter holds. For $\theta>0$ let us introduce the auxiliary problem with value function
\[
v_\theta(t,x)=\sup_{0\le \tau\le T+\theta-t}\E_x\Big[h(X_{\tau\wedge\tau^\dagger})-\int_0^{\tau\wedge\tau^\dagger}\Big(p+q e^{-X_s}\Big)\ud s\Big].
\]
Since $h$ and the process $X$ are time-homogeneous we clearly have $v_\theta(t,x)=v(t-\theta,x)$ for all $t\in[\theta,T]$ and $x\in\R_+$. In particular, since we are assuming $c(x)>0$ for $x\ge \hat x_2$, the optimal stopping boundary for the auxiliary problem is $c_\theta(x)=\theta+c(x)$ for $x\ge \hat x_2$. Hence, $c_\theta(x)\ge \theta$ for all $x\ge \hat x_2$. Then, by the same argument above with $v_\theta(0,x)$ instead of $v(0,x)$ we reach again a contradiction.
\end{proof}

\begin{figure}[t!]
	\centering
	\begin{center}
		\hspace*{-0.7cm}
		\includegraphics[scale=0.4]{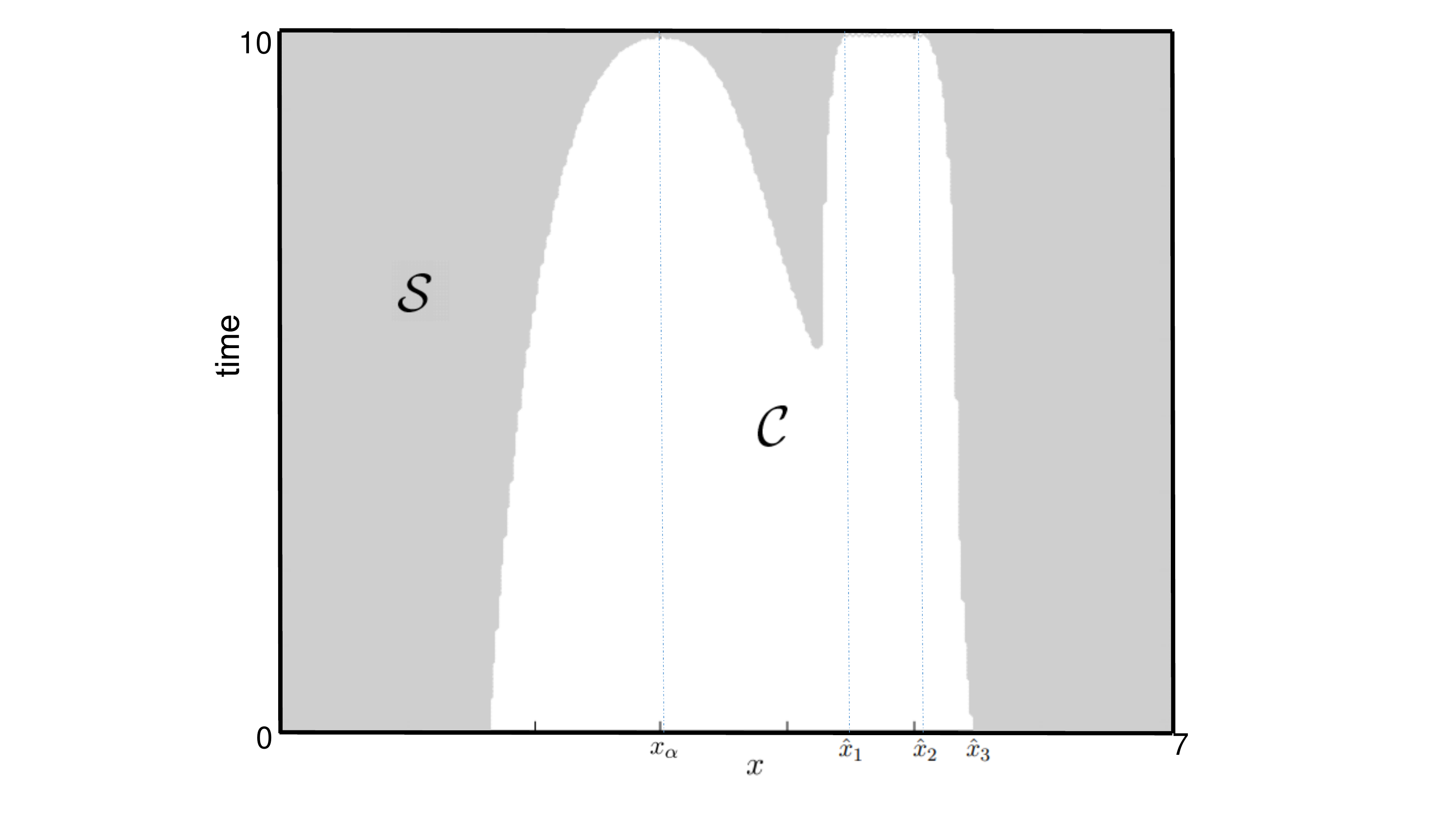}
		\vspace{-0.5cm}
	\end{center}
	\caption{The optimal surrender region and boundary in Case (ii) of management fees (see Proposition \ref{prop:MF}).}
	\label{fig4}
\end{figure}
Case (i) can be formally interpreted as the limiting situation of Case (ii) with $\hat x_1=\infty$. By the same arguments above we obtain the analogue of (i) and (ii) in Proposition \ref{prop:monot-c2} but with $c(x)=0$ for $x\ge x_2$ therein and $c(x_2-)=0$.  Once again continuity of the boundary follows from Theorem \ref{prop:cont}.

The presence of management fees paid at a constant (proportional) rate, reduces the incentive created by the bonus mechanisms in the policy. That is why we observe an upper exercise boundary at which the policyholder surrenders the contract when the expected gains from staying in the contract (and collecting the bonus rate) are outweighed by the expected cost of future management fees.



\section*{Appendix}
\begin{proof}[{\bf Proof of Lemma \ref{lem:tech}}] Here we borrow arguments from the proof of \cite[Thm.~1]{DeAK18}. To keep a simple notation, in what follows we set $X_t=X^{x_\alpha}_t$ everywhere. From the equality
	\begin{align*}
	|X_t-x_\alpha|=\int_{0}^{t}  \textrm{sign}(X_s -x_\alpha)\ud X_s+L^{x_\alpha}_{t},
	\end{align*}
	we deduce that
	\begin{align*}
	L^{x_\alpha}_{s\wedge \rho_\varepsilon}&=|X_{s\wedge \rho_\varepsilon}-x_\alpha|-\int_{0}^{s\wedge \rho_\varepsilon}  \textrm{sign}(X_u -x_\alpha)\ud X_u\\
	&=|X_{s\wedge \rho_\varepsilon}-x_\alpha|-\int_{0}^{s\wedge \rho_\varepsilon}  \textrm{sign}(X_u -x_\alpha)\sigma \ud W_u\\
&\qquad\qquad\qquad\:\:\:\,-\int_{0}^{s\wedge \rho_\varepsilon}  \textrm{sign}(X_u -x_\alpha)\pi(X_u) \ud u.
	\end{align*}
Notice that $\pi(\cdot)$ is bounded on $[x_\alpha-\varepsilon,x_\alpha+\varepsilon]$ by a constant $c_{\pi,\eps}>0$. Moreover, since $|X_{s\wedge \rho_\varepsilon}-x_\alpha| \leq \varepsilon$, then $1 \ge \eps^{-p}|X_{s\wedge \rho_\varepsilon}-x_\alpha|^p$ for $p>0$. Taking expectation and using these two observations we get
\begin{align*}
\E_{x_\alpha} \left[  L^{x_\alpha}_{s\wedge \rho_\varepsilon}\right]
= &\, \E_{x_\alpha} \left[|X_{s\wedge \rho_\varepsilon}-x_\alpha|\right]-\E_{x_\alpha} \left[\int_{0}^{s\wedge \rho_\varepsilon}  \textrm{sign}(X_u -x_\alpha)\,\pi(X_u)\ud u\right] \nonumber\\
\ge &\, \E_{x_\alpha} \left[|X_{s\wedge \rho_\varepsilon}-x_\alpha|\right]- \ c_{\pi,\varepsilon}\E_{x_\alpha} [s\wedge \rho_\varepsilon]\nonumber\\
\ge &\,  \varepsilon^{-p}\,\E_{x_\alpha} \left[|X_{s\wedge \rho_\varepsilon}-x_\alpha|^{1+p}\right]- c_{\pi,\varepsilon} \ s. \nonumber
\end{align*}
The expectation of the absolute value above can be estimated using the integral form of the SDE for $X$ and the inequality $|a+b|^{1+p} \ge \frac{1}{2^{1+p}}|a|^{1+p}-|b|^{1+p}$, for all $a,b \in \R$ (see, \cite[Ch.~8, Sec.~50, p.~83]{KF}). That is, 	
\begin{align} \label{EL}	
\E_{x_\alpha} &\left[  L^{x_\alpha}_{s\wedge \rho_\varepsilon}\right]\notag\\
\ge &\,  \varepsilon^{-p}\,\E_{x_\alpha} \left[\left|\int_{0}^{s\wedge \rho_\varepsilon} \pi(X_u)\ud u+\sigma W_{s\wedge \rho_\varepsilon}\right|^{1+p}\right]- \ c_{\pi,\varepsilon}\ s\nonumber \\
\ge &\,  \varepsilon^{-p}\left\{\left(\frac{\sigma}{2}\right)^{1+p}\E_{x_\alpha} \left[\left|W_{s\wedge \rho_\varepsilon}\right|^{1+p}\right] -\E_{x_\alpha} \left[\left|\int_{0}^{s\wedge \rho_\varepsilon} \pi(X_u)\ud u\right|^{1+p}\right] \right\}- \ c_{\pi,\varepsilon}\ s \\
\ge &\,  \varepsilon^{-p}\left\{\left(\frac{\sigma}{2}\right)^{1+p}\E_{x_\alpha} \Big[|W_{s\wedge \rho_\varepsilon}|^{1+p}\Big] - c_{\pi,\varepsilon}^{1+p} \ s^{1+p}\right\}- \ c_{\pi,\varepsilon}\ s.\nonumber
\end{align}
Now Burkholder-Davis-Gundy inequality and Doob's inequality imply that there exists a positive constant $c_p$ such that
\begin{align*}
\E_{x_\alpha} \Big[|W_{s\wedge \rho_\varepsilon}|^{1+p}\Big] \ge &\, c_p\, \E_{x_\alpha} \left[ \left(s\wedge \rho_\varepsilon\right)^{\frac{1+p}{2}} \right]\ge  c_p\, \E_{x_\alpha} \left[\mathds{1}_{\{s<\rho_\varepsilon\}} \right]s^{\frac{1+p}{2}}\\
= &\, c_p\, s^{\frac{1+p}{2}}\big( 1 -\P_{x_\alpha} (\rho_\varepsilon \leq s)\big).
\end{align*}
Inserting the last inequality in \eqref{EL} we get
\begin{align*}
\E_{x_\alpha}& \left[L^{x_\alpha}_{s\wedge \rho_\varepsilon}\right]\ge  c_1 \ s^{\frac{1+p}{2}}-c_2 \left(s+s^{1+p}+\P_{x_\alpha} \left(\rho_\varepsilon \leq s\right) s^{\frac{1+p}{2}}\right)=:\lambda(s),
\end{align*}
for some suitable positive constants $c_1=c_1(\varepsilon,p)$ and $c_2=c_2(\varepsilon,p)$.	Since $p\in(0,1)$, in the limit as $s \downarrow 0$ we get
\[
\lambda(s) = c_1 s^{\frac{1+p}{2}}\left[1-\tfrac{c_2}{c_1}\big(\P_{x_\alpha} \left(\rho_\varepsilon \leq s\right)+s^{\frac{1-p}{2}}+s^{\frac{1+p}{2}}\big)\right]\ge \tfrac{1}{2}c_1 s^{\frac{1+p}{2}}
\]
because $\P_{x_\alpha} \left(\rho_\varepsilon \leq s\right)+s^{\frac{1-p}{2}}+s^{\frac{1+p}{2}}\downarrow 0$ as $s\to0$. Given that $\tfrac{1+p}{2}\in(\tfrac{1}{2},1)$ we have $\lambda(s)>\ell \cdot s$ as $s\downarrow 0$ for any constant $\ell>0$. That implies the claim of the lemma.
\end{proof}

\vspace{+15pt}

\noindent{\bf Acknowledgment}: The authors were financially supported by Sapienza University of Rome, research project
``\emph{Life market: a renewal boost for quantitative management of longevity and lapse risks}'', grant no.~RM11916B8953F292. T.~De Angelis gratefully acknowledges support via EPSRC grant EP/R021201/1, ``\emph{A probabilistic toolkit to study regularity of free boundaries in stochastic optimal control}''.

\end{document}